\documentclass[11pt, onecolumn, journal,final]{IEEEtran}

\usepackage{enumerate,latexsym,amsmath,amssymb,amsfonts,verbatim}
\usepackage{dsfont}
\usepackage[final]{graphicx}
\usepackage{overpic}
\usepackage{url}
\usepackage{multirow}
\usepackage{array}
\usepackage{subfigure}
\usepackage{hyperref} 
\usepackage{mathrsfs}
\usepackage{cite}
\usepackage{setspace}
\usepackage{bbm}

\newtheorem{theorem}{Theorem}[section]
\newtheorem{lemma}[theorem]{Lemma}

\newtheorem{corollary}[theorem]{Corollary}
\newtheorem{definition}{Definition}

\begin{document}
\doublespacing

\title{Quantifying the Information Leakage in Timing Side Channels in Deterministic Work-Conserving Schedulers (Draft)}
\author{
\IEEEauthorblockN{Xun~Gong,~\IEEEmembership{Student~Member,~IEEE,}
        and~Negar~Kiyavash,~\IEEEmembership{Senior Member,~IEEE}}\\
 \thanks{This work was supported in part by National Science Foundation through the grant CCF 10-65022, CCF 10-54937 CAR, and in part by Air Force through the grant FA9550-11-1-0016, FA9550-10-1-0573.
 
 X. Gong is with the Coordinated Science Laboratory and the Department of Electrical Engineering, University of Illinois at Urbana-Champaign, Urbana, IL 61801 USA (email: \url{xungong1@illinois.edu})
 N. Kiyavash is with  the Coordinated Science Laboratory  and the Department of Industrial and Enterprise Systems Engineering, University of Illinois at Urbana-Champaign, Urbana, IL 61801 USA (email: \url{kiyavash@illinois.edu})
 }   
}

\maketitle

\begin{abstract}
When multiple job processes are served by  a single scheduler, 
the queueing delays of one process are often affected by the others,
resulting in  a timing side channel that leaks the arrival pattern of  one process to the others. 
In this work, we study such a timing side channel between 
a regular user and a malicious attacker. 
Utilizing Shannon's mutual information as a measure of information leakage between
the user and attacker, we analyze privacy-preserving behaviors of common work-conserving schedulers.
We find that the attacker can always learn perfectly the user's arrival process in a longest-queue-first (LQF) scheduler.
When the user's job arrival rate is very low (near zero),  first-come-first-serve (FCFS) and round robin schedulers both completely reveal the user's arrival pattern. 
The near-complete information leakage  in the low-rate traffic region is proven to be reduced by half in a work-conserving version of TDMA (WC-TDMA) scheduler, 
which turns out to be privacy-optimal in the class of deterministic-working-conserving (det-WC) schedulers, according to a universal lower bound on information leakage we derive for all det-WC schedulers. 
\end{abstract}

\IEEEpeerreviewmaketitle

\section{Introduction}
\PARstart{I}t has long been known that event times could be used for covert communication~\cite{lampson73:acm}. 
For instance, by encoding messages in transmission  times of events, 
 an event scheduler can create
 a {\em timing covert channel} to 
 any observer that sees the time events occur. 
 Some notable timing covert channels include the CPU scheduling channel~\cite{millen89}, in which 
one process encodes a message into sizes of the jobs it hands to a CPU shared with another process which decodes this information through monitoring CPU's busy period, and the IP timing channel~\cite{Cabuk2009}, in which messages are embedded in the inter-arrival-times of packets.

More recently,  it has  been shown that event times also incidentally leak information, resulting in  {\em timing side channels}.  Unlike a covert channel,  there is no active message sender in a side channel. Instead, an attacker infers information about the other users  from the timing evidence left on a shared resource. 
Such a timing side channel exists between two users sending jobs to the same queue.
The queuing delays of one user's jobs convey information about the activities of the other.  
Timing side channels have been previously exploited to
learn the activities of cloud clients and home broadband customers~\cite{ristenpart09:ccs, gong2012pets}. 
In cloud computing infrastructures, such as  Amazon  Elastic Compute Cloud (EC2),   a server often hosts jobs from multiple clients. This provides a malicious client with the opportunity to probe workloads of his cloud neighbors~\cite{ristenpart09:ccs}.
Likewise, a timing side channel can be built  in the home digital subscriber line (DSL) router. 
An attacker pinging the DSL user may learn the user's web traffic pattern  because the pings and the user's packets share the downstream queue at the DSL router~\cite{gong2012pets}.

In this paper, we study a  timing side channel that arises when two users share a joint event queue.  One user is assumed to be a malicious attacker 
 who wants to learn the other user's job arrival pattern based on  the delays his jobs experience.
The amount of coupling between the user and attacker's jobs largely depends on
the scheduling policy of the job server. 
A scheduler certainly can  eliminate  the side channel by applying  a time division multiple access (TDMA) policy,
which decouples 
 service to the users but adds unnecessary delays.
 On the other hand, in  work-conserving schedulers, which  achieves delay optimality by keeping busy as long as the queue is not empty,
 timing side channels are inevitable. 
 Kadloor et al.\ \cite{kadloor13} characterized the information leakage of work-conserving schedulers for an attacker that could issue
infinitesimally small jobs. 
This raises the question: {\em Could side channel information leakage be alleviated if 
the attacker is not allowed to issue jobs with arbitrarily small sizes?} 
In fact,  in many real systems, there are requirements on acceptable job sizes. 
For instance, the limit on 
  network packet sizes  often pre-fixes~\cite{peterson2007computer}.
We answer this question by considering a  scenario, where users are required to send jobs of comparable sizes. 
Additionally, we measure the leakage of a scheduler in terms of  performance of the best attacker who aims to learn the exact arrival times of the user's jobs.
This is a departure from~\cite{kadloor13}, where the attacker's goal was to learn the counts of the user's jobs in each clock period. The current metric  captures loss in privacy of the user more accurately. 
 Our main contributions of this work are summarized in the following:
\begin{itemize}
\item We develop an information-theoretic framework to analyze  timing side channels in job schedulers. 
Considering a scheduler serving a user and an attacker, we measure the information leakage  using Shannon's mutual information between the user's job arrival process and the attacker's job arrival and departure processes. 

\item We demonstrate that most commonly deployed work-conserving scheduling policies are not privacy optimal: the longest-queue-first~(LQF) scheduler leaks the user's arrival pattern completely;
when the user's job arrival rate is near zero, both first-come-first-serve (FCFS) and round robin schedulers completely reveal the user's arrival pattern, while a work-conserving TDMA-like scheduler leaks the user's arrival process  half of the time. 

\item We derive a lower-bound on information leakage for  all deterministic work-conserving (det-WC) schedulers, where the server takes deterministic actions and stays busy as long as  there are jobs to serve. The lower bound shows that  in the low-rate traffic region, the attacker learns the user's arrival pattern for at least half of the time.
The implication of this study is that deploying det-WC schedulers in applications, in which privacy is a concern, is a poor choice.  
\end{itemize}

\section{Related Work}

Traditionally, timing channels are for the most part studied in the context of covert communication.
Most of the literature focuses on the capacity of such channels.
Anantharam and Verd\'{u}~\cite{AnantharamVerdu96} studied the timing channel 
between the arrival and departure process of a single user $\cdot/G/1$ queue, and showed 
that capacity is minimized when  the service times of jobs are exponentially distributed. 
For such a queue, bounds on the capacity for Bounded Service Timing Channels (BSTC), in which the service time distributions have bounded support, were derived in~\cite{sellke2007capacity}.
Riedl et. al\ \cite{riedl2011finite} considered the usage of the same channel with finite-length codewords, and obtained a lower bound on the maximal rate achievable. 
A covert channel between two job processes sharing a round robin scheduler was studied in~\cite{millen89}.  Assuming all jobs have the same size, it was proved that the channel capacity  is
$\log\left(\frac{1+\sqrt{5}}{2}\right)$ bits per time slot.
Strategies for mitigating timing covert channels were studied in~\cite{Siva, Giles02, askarov10:ccs, zhang11:ccs}. 
The main proposed countermeasure  idea is 
to weaken the correlation between event times seen by the sender and receiver 
via injecting `dummy' delays. 

On the application side, 
timing side channels were  exploited  in network traffic analysis to compromise user anonymity. In~\cite{Murdoch2005},  round-trip times~(RTTs) of probe packets sent to routers were measured to estimate available bandwidths at the router, which were subsequently used to expose the identity of relays participating in a circuit of  the anonymous communication networks, such as  Tor~\cite{Dingledine2004} or MorphMix~\cite{morphmix}.
In~\cite{kiyavash2013timing},
Kiyavash et al.\ designed and implemented a spyware communication circuit in the widely used carrier sense multiple access with collision avoidance (CSMA/CA) protocol, using the timing channel  resulting from transmission of packets.
In~\cite{gong2012pets} and~\cite{Sachin10},  it was shown that an attacker can create a timing side channel inside a DSL router using frequent pings, and recover DSL user's traffic pattern from monitored  RTTs. 
Queuing side channels in shared queues were analyzed in~\cite{kadloor13} and~\cite{gong2011isit},
where the information leakage was measured by minimum-mean-square-error and equivocation, respectively. 
With the goal of measuring the number of jobs  from the user in a clock period, it was shown in both~\cite{kadloor13} and~\cite{gong2011isit} that an FCFS scheduler completely leaked the user's traffic pattern if the attacker could send at least one job in every clock period.  Additionally, assuming the attacker was able to issue jobs with infinitesimal  sizes,  it was proven in~\cite{kadloor13} that 
round robin is privacy optimal among all work-conserving schedulers; yet it leaks substantial  information about the user. 


\section{Problem Formulation}

In this section, we introduce the notation and system model. Throughout this section: bold script $\mathbf{A}$ denotes the infinite sequence  $\{A_1, A_2, \cdots\}$,  $\mathbf{A}^n$ denotes the finite sequence  $\{A_1, A_2, \cdots, A_n\}$, and $\mathbf{A}^{j}_{i}$ denotes the subsequence $\{A_{i}, A_{i+1}, \cdots, A_{j}\}$, where $j\geq i$. 
\subsection{System Model}
\label{sec:sys_mod}
We consider the timing side channel in a scheduler processing jobs from a regular user and a malicious attacker in discrete time, as depicted by~Figure~\ref{fig:sys_mod}. 
In each time slot, the user (and the attacker) either issues one job or remains idle. All jobs have the same size and take one time slot to service. The user sends jobs according to a Bernoulli process with rate $\lambda$.
The attacker, who wants to infer the user's arrival times, picks time slots according to his attack strategy and sends jobs with the long term rate  $\omega$,
not exceeding $1-\lambda$, in order to preserve the queue's stability. 
We assume all arrival and departure events occur at the beginning of time slots.

 \begin{figure}[t]
   \centering
   \includegraphics[width=0.6\columnwidth]{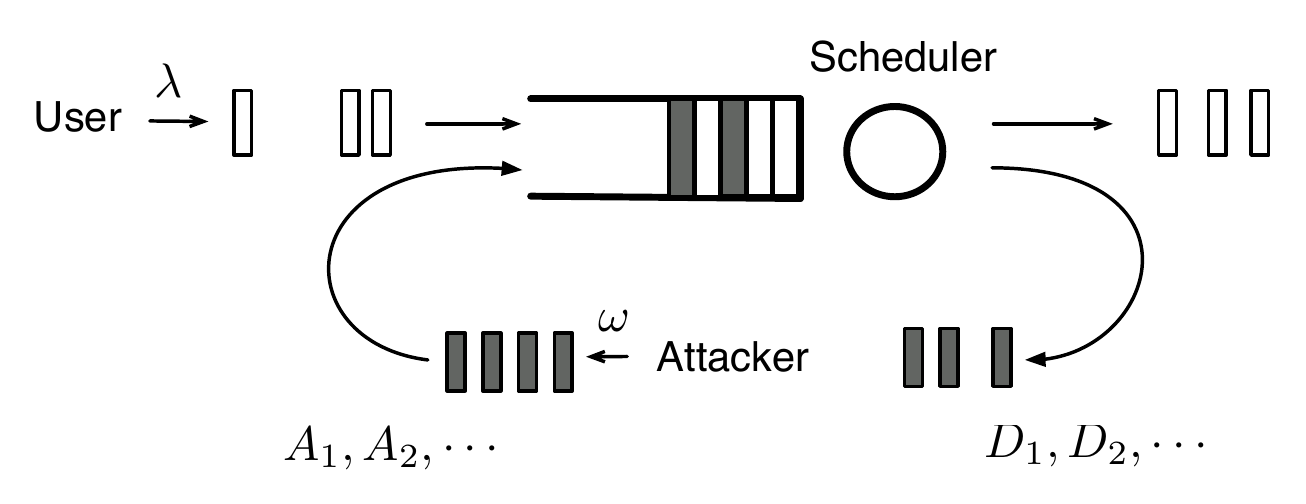} 
   \caption{
  A scheduler services jobs from two arrival processes; one  from a malicious attacker
(solid) and one from a regular user (blank). The attacker sends
jobs to the scheduler to sample the queue status, based on which he learns about the arrival pattern of the other user.}
   \label{fig:sys_mod}
\end{figure}

\subsection{Information Leakage Metric}

We measure  information leakage in this timing side channel by  Shannon's mutual information between the user's arrival process and the attacker's observations, comprised of  his arrival and departure times. 
Similar metrics, e.g., Shannon's equivocation,  have been frequently used for quantifying information leakage in communication systems, such as the wiretap channel~\cite{Wyner75}.  Denote the arrival event sequence of the user's jobs in each time slot  by $\boldsymbol{\delta}=\{\delta_1,\delta_2,\cdots\}$, where $\delta_i\sim Bernoulli(\lambda), i \in \mathbb{Z}$, and denote the arrival and departure times of the attacker's jobs by $\mathbf{A}=\{A_1,A_2,\cdots\}$ and $\mathbf{D}=\{D_1,D_2,\cdots\}$, respectively. 
\begin{definition}
The information leakage of a timing side channel in a queue shared by a user and an attacker is defined as
\begin{equation}
\begin{aligned}
\mathcal{L}{(\lambda)}&=\underset{\mathbf{A}:\underset{k\to\infty}{\lim}\frac{k}{A_k}<1-\lambda}{\max}\text{ }
\underset{n\to\infty}{\lim}\frac{I\left(\boldsymbol{\delta}^{n}; \mathbf{A}^{m},\mathbf{D}^{m}\right)}{n},
\end{aligned}\label{eq:privacy}
\end{equation}
where $I\left(\cdot; \cdot\right)$ denotes Shannon's mutual information,  $\lambda$ is the user's arrival rate, 
and $m$ is the number of jobs the attacker has issued by time $n$:
\begin{equation}m=\underset{}{\sup}\{k:A_k\leq n\}.
\label{eq:m}\end{equation}
\label{def:privacy}
\end{definition}

The  leakage $\mathcal{L}$ characterizes  the information gain of the attacker deploying the best possible attack strategy satisfying the rate restriction. 
Hence, a larger leakage $\mathcal{L}$ signifies a larger compromise in the user's privacy
 through the timing side channel. 
 Let $H(\lambda)$ denote the entropy rate of the user's arrival process, which is assumed to be Bernoulli. 
\begin{definition}
The information leakage ratio  of a timing side channel in a queue shared by a user and an attacker  is defined as
\begin{equation}
\begin{aligned}
\mathcal{R}(\lambda)&=\underset{\mathbf{A}:\underset{k\to\infty}{\lim}\frac{k}{A_k}<1-\lambda}{\max}\text{ }
\underset{n\to\infty}{\lim}\frac{I\left(\boldsymbol{\delta}^{n}; \mathbf{A}^{m},\mathbf{D}^{m}\right)}{n H\left(\lambda\right)}.
\end{aligned}\label{eq:privacy_ratio}
\end{equation}
\label{def:privacy}
\end{definition}

The information leakage ratio $\mathcal{R}$ is a better metric for comparing the leakage across users with various rates.  
The value of $\mathcal{L}$ (and $\mathcal{R}$) clearly depends on the scheduling policy. 
For instance,  for the TDMA policy,  in which fixed time slots are preassigned for serving each arrival process, both $\mathcal{L}$ and $\mathcal{R}$ are zero. This is because service times of the attacker's jobs are statistically independent of the user's arrival pattern.  
Unfortunately, TDMA  is wasteful  and adds significant delays by  causing the scheduler to idle.  
Therefore, such complete isolation of users' job processes is often not desired in practice.
In this work, we analyze the information leakage of timing channels in work-conserving (delay-optimal) schedulers and  
investigate whether good policies that are simultaneously privacy and delay optimal exist. 


\section{Information Leakage in  Deterministic Work-Conserving Schedulers}

In this section, we characterize or derive bounds on the leakage in the class of  deterministic-work-conserving (det-WC) schedulers. These schedulers service jobs in a deterministic fashion and do not idle as long as there are jobs in the queue. 
Our main results are summarized in Figure~\ref{fig:it_summary}.

We show that even when the attacker is required to send jobs of a comparable size to the user, all det-WC schedulers leak at least half of the user's traffic pattern in the low-rate region. This is proved by deriving a universal lower bound for all det-WC schedulers as depicted in Figure~\ref{fig:it_summary}. 

The attacker learns  completely the arrival process of the user in an LQF scheduler by simply maintaining his own queue length  at one (the flat solid line at the top in  Figure~\ref{fig:it_summary}). In an FCFS scheduler, instead of the exact arrival times, the attacker can infer the number of jobs (arrivals) of the user between  any of his two consecutive  jobs by sampling the queue at his maximum rate. This leads to a severe leakage in the arrival process of a user with low arrival rate (the blue solid curve in  Figure~\ref{fig:it_summary}), as the attacker can sample the queue frequently enough to obtain an accurate estimate of the user's arrival event in each time slot. 
We derive a lower bound on the privacy leakage for round robin (the green dashed curve in Figure~\ref{fig:it_summary}), which can be achieved by an attacker who issues a new job right after his previous job is serviced. This attack strategy allows the attacker to detect when the user's queue gets empty.
As depicted in  Figure~\ref{fig:it_summary},  it  provides sufficient information about the user's arrival pattern at the low-rate region, resulting in a complete information leakage when the user's arrival rate $\lambda\to0$.

 \begin{figure}[t]
   \centering
   \includegraphics[width=0.75\columnwidth]{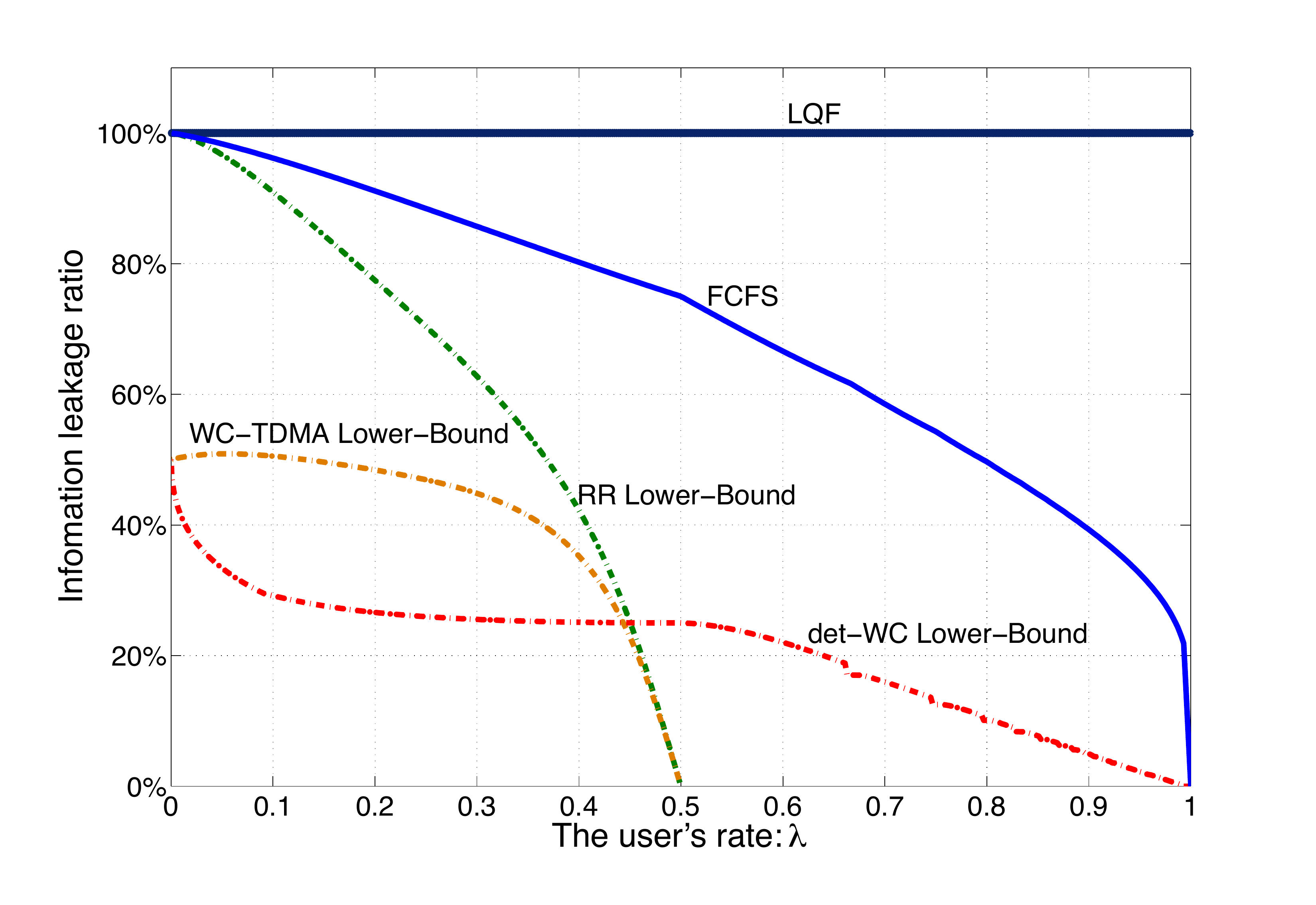} 
   \caption{
Information leakage ratios in deterministic-work-conserving schedulers. In LQF,  the user's arrival process is completely leaked to the attacker, which also occurs in the FCFS and round robin scheduler when the user's rate is very low. A work-conserving TDMA scheduler reduces the fraction of leaked information in the low-rate region by half; the lower-bound on the leakage of WC-TDMA is tight at $\lambda\to 0$.}  
 \label{fig:it_summary}
\end{figure}  

The near-complete information leakage  in the low-rate region is alleviated by the work-conserving TDMA (WC-TDMA) scheduler. 
Like TDMA, the WC-TDMA scheduler reserves slots for each arrival process; e.g., odd slots for the user and even ones for the attacker.
However, in each time slot, if the preassigned 
user has no jobs, the scheduler serves the other user with jobs waiting for service.
Such work-conserving behavior enables the attacker to correctly detect  arrivals on time slots reserved by the user.
We derive a lower bound on the leakage of a WC-TDMA scheduler, which is proved to be tight when the user's rate $\lambda\to0$ (the orange dashed curve in Figure~\ref{fig:it_summary}). 
This means the attacker can learn  the  arrival pattern of a low-rate user perfectly for half of the time, and further implies that for $\lambda\to0$, WC-TDMA is a privacy optimal policy in the det-WC class 
as it meets  the det-WC universal lower-bound.

\subsection{Longest-Queue-First}
\label{sec:lqf}
We first analyze the leakage of an LQF
 scheduler, which we can exactly characterize.  In each time slot, the LQF scheduler services the first buffered job from the user that has more jobs queued up so far. In the case of a tie, the scheduler serves a predetermined user first. 

Since  the LQF scheduler takes actions by comparing queue lengths of users, a smart attacker can accurately 
learn the change in the user's queue state by maintaining 
 his queue length constantly at one.
 Assuming the user has priority of service at a tie, such an attacker always knows whenever the user's queue size passes 0 and detects every job sent by the user, as further explained below. 

\begin{theorem}
The information leakage of an LQF scheduler serving a user and an attacker is  given by
\begin{equation}
\begin{aligned}
\mathcal{L}_{LQF}(\lambda)= H(\lambda),
\end{aligned}
\label{eq:lqf_main}
\end{equation}
where  $\lambda$ is the user's arrival rate. Consequently, 
$\mathcal{R}_{LQF} (\lambda)=1$, for all $0<\lambda<1$.
\label{eq:lqf_main_1}
\label{theo:lqf_main}
\end{theorem}
\begin{proof}
Consider a {\em nonstop monitoring} attack strategy (Figure~\ref{fig:nonstop_attack}), where the attacker issues a new job immediately after his previous is serviced.
Recall in our model, all arrival and departure events happen at the beginning of time slots. 
Thus, in the nonstop monitoring attack, we have
\begin{equation}
\begin{aligned}
{A}_k={D}_{k-1}, \quad k\in\mathbb{Z}.
\label{eq:nonstop}
\end{aligned}
\end{equation} 

Such an attacker always has a single job in the queue. 
Assume the user gets served first when a tie  happens.
Then, the scheduler never serves the attacker unless  the user has no job left. 
As a result, whenever  the user issues a new job,  the attacker experiences a time slot of delay, i.e.,
\begin{equation}
\begin{aligned}
\delta_i=\begin{cases} 0 & \mbox{if  } \exists  k\in\mathbb{Z}, \mbox{ s.t. }   {D}_k = i+1   \\ 1 & \text{otherwise}  \end{cases},
\quad i \in\mathbb{Z}.
\end{aligned}
\label{eq:lqf_q}
\end{equation}
Similarly, if the attacker has priority of service when there is a tie in queue lengths, he would get served if and only if the user's queue length
falls below 2, in which case \eqref{eq:lqf_q} also holds.

\begin{figure}[t]
   \centering
   \includegraphics[width=0.6\columnwidth]{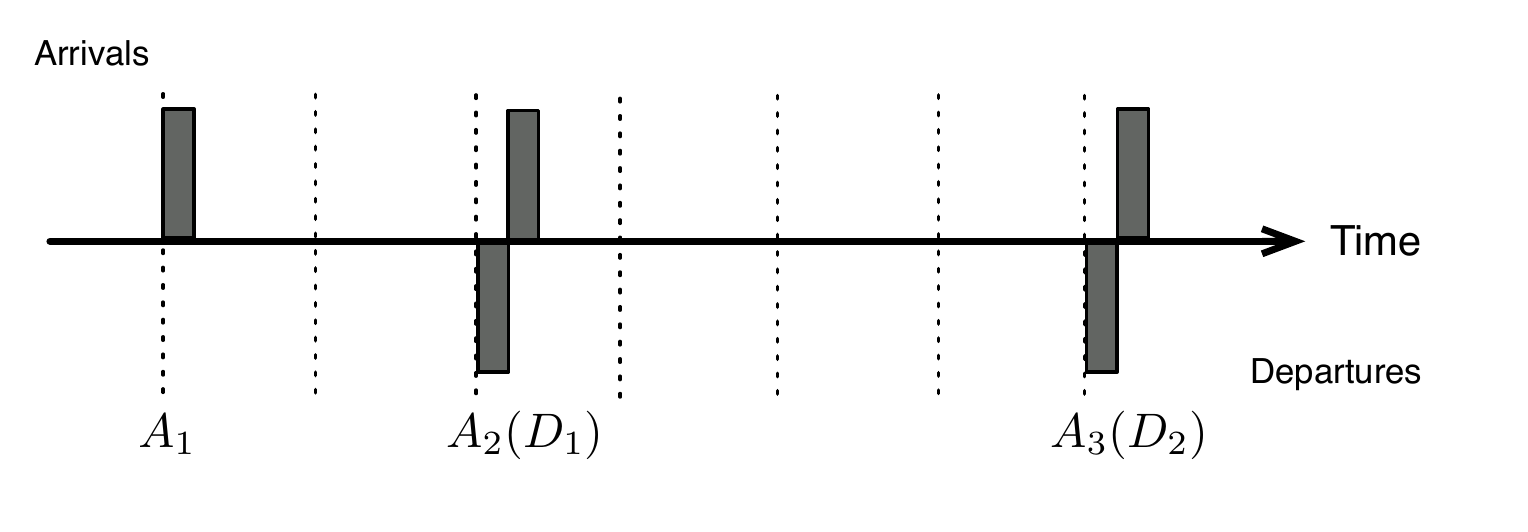}
   \caption{Nonstop monitoring: the attacker issues a new job right after a previous job departs from the queue. }
   \label{fig:nonstop_attack}
\end{figure}

Therefore,  we can obtain a lower bound on information leakage which results from the nonstop-monitoring attack  as  follows:
\begin{equation}
\begin{aligned}
\mathcal{L}_{LQF} (\lambda) &\geq
\underset{n\to\infty}{\lim}\frac{I\left(\boldsymbol{\delta}^{n}; {\mathbf{A}}^{{m}},{\mathbf{D}}^{{m}}\right)}{n}\\
&\overset{(a)}{=}  \underset{n\to\infty}{\lim}\frac{I\left(\boldsymbol{\delta}^n;{\mathbf{D}}^{{{m}}}\right)}{n}\\
&= H(\lambda)-\underset{n\to\infty}{\lim}\frac{H\left(\boldsymbol{\delta}^{n}\big| {\mathbf{D}}^{{{m}}}\right)}{n}\\
&\overset{(b)}{=}  H(\lambda).
\end{aligned}
\label{eq:lqf_lb}
\end{equation}
where   $(a)$ and $(b)$ follow from~\eqref{eq:nonstop} and~\eqref{eq:lqf_q}, respectively.  
Additionally, since the leakage is always upper-bounded by the total entropy rate of the user's arrival process, $H(\lambda)$,  we have $\mathcal{L}_{LQF}(\lambda)=H(\lambda)$.
\end{proof}

LQF is a low-complexity approximation of the MaxWeight scheduling, a  throughput-optimal algorithm applied in network switches~\cite{shakkottai02}. It requires less buffer storage than other common scheduling algorithms,  such as FCFS and round robin~\cite{gail93}. However, as seen in Theorem~\ref{theo:lqf_main}, LQF fully exposes arrival pattern of the user to an attacker.

\subsection{First-Come-First-Serve}
\label{subsec:fcfs}

We subsequently analyze the leakage of FCFS, a simple service policy widely applied in network systems. 
At each time slot, the FCFS scheduler services the job at the head of the queue.\footnote{For the sake of convenience, we assume that when both the user and the attacker issue a job in one time slot, the attacker's job enters the queue first.}
FCFS reveals the queue length $q(\cdot)$ of the buffer to an attacker through queueing delays of his jobs because
\begin{equation}
\begin{aligned}
q\left(A_k\right)=D_k-A_k-1,  \quad k\in\mathbb{Z},
\label{eq:queue_service}
\end{aligned}
\end{equation}
where `1' accounts for the service time of the $k^{th}$ attacker's job. 
As a result, the attacker can  frequently sample the state of the buffer queue and then estimate the number of the user's 
arrivals.

Let $N(t)$ denote the counting function associated with the user's arrivals at time $t$, then
\begin{equation}
N(t)=\sum_{j=1}^{t}\delta_j, \quad t\in \mathbb{Z}.
\label{eq:count_func}
\end{equation}
The following theorem on optimal sampling of a Bernoulli processes is necessary for proving our main result on information leakage of the FCFS scheduler. 
For the ease of notation, we define in the following
\begin{equation}
 \alpha_\epsilon\overset{\bigtriangleup}{=}\frac{\left\lceil \frac{1}{\epsilon} \right\rceil-\frac{1}{\epsilon}}{\left\lceil \frac{1}{\epsilon} \right\rceil-\left\lfloor \frac{1}{\epsilon} \right\rfloor}, \quad \forall \text{ }0<\epsilon<1.
 \label{eq:alpha}
\end{equation}

\begin{figure}[t]
   \centering
   \includegraphics[width=0.6\columnwidth]{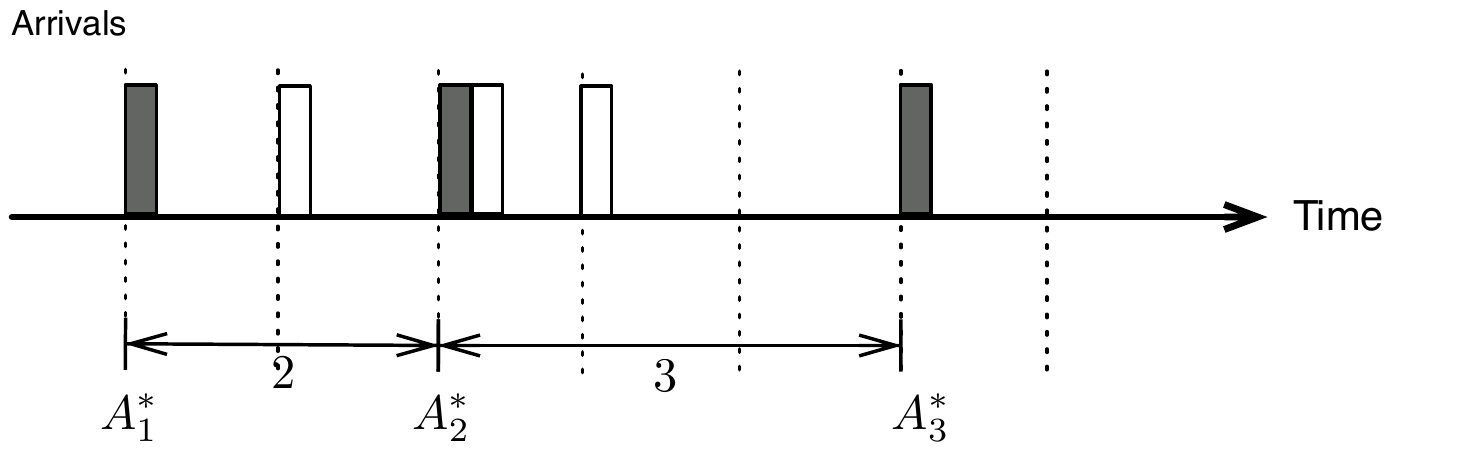} 
   \caption{{Periodic sampling:} given a sampling rate $\omega$, the attacker issues jobs (solid) periodically, with inter-arrival times chosen from $\left\lfloor \frac{1}{\omega}\right\rfloor$ and $\left\lceil \frac{1}{\omega}\right\rceil$. For example, if $\omega=0.4$, the inter-arrival time ${A}^{*}_i-{A}^{*}_{i-1}$ would take value of 2 and 3 with equal probability.}
   \label{fig:attack_strategy}
\end{figure}

\begin{theorem}
Consider sampling a Bernoulli arrival process $\boldsymbol{\delta}=\{\delta_1,\delta_2, \cdots\}$ at times $\mathbf{A}=\{A_1,A_2,\cdots\}$.
For a fixed sampling rate $\omega=\underset{k\to\infty}{\lim}\frac{k}{A_k}$, the following periodic sampling strategy (Figure~\ref{fig:attack_strategy}) is optimal:
\begin{equation}
\begin{aligned}
A^{*}_{k}-A^{*}_{k-1}=\begin{cases} \left\lfloor\frac{1}{\omega}\right\rfloor & \mbox{w.p. } \alpha_\omega \\ \left\lceil \frac{1}{\omega} \right\rceil & \mbox{w.p. } 1-\alpha_{\omega}  \end{cases}, \quad k\in \mathbb{Z}.
\end{aligned}
\label{eq:attack_strategy_1}
\end{equation} 
The optimality is defined in the sense of maximizing  the entropy rate of the sampled process.
This optimal value is given by
\begin{equation}
\underset{k\to\infty}{\lim}\frac{H\left(N\left(A^{*}_1\right),\cdots,N\left(A^{*}_k\right)\right)}{k}= \alpha_\omega  H\left(\sum_{i=1}^{\left\lfloor \frac{1}{\omega} \right\rfloor}\delta_i\right)  +(1-\alpha_\omega) H\left(\sum_{i=1}^{\left\lceil \frac{1}{\omega} \right\rceil}\delta_i\right).
\label{eq:periodic_sampling_main}
\end{equation}
\label{theorem:periodic_sampling}
\end{theorem}
\begin{proof}
See Appendix~\ref{app:opt_samp}.
\end{proof}

\begin{theorem}
The information leakage of an FCFS scheduler serving a user and an attacker is given by
\begin{equation}
\begin{aligned}
\mathcal{L}_{FCFS}(\lambda) = (1-\lambda)\left( \alpha_{1-\lambda}
{H\left(\sum_{i=1}^{\left\lfloor \frac{1}{1-\lambda} \right\rfloor}\delta_i\right)} + (1-\alpha_{1-\lambda}){H\left(\sum_{i=1}^{\left\lceil \frac{1}{1-\lambda} \right\rceil}\delta_i\right)}\right),
 \end{aligned}
\label{eq:fcfs_main}
\end{equation}
where  $\lambda$ is the user's arrival rate, and $\delta_i$'s are i.i.d. $Bernoulli(\lambda)$ random variables.
If $\frac{1}{1-\lambda}\in\mathbb{Z}$, \eqref{eq:fcfs_main} is simplified to
$\mathcal{L}_{FCFS}(\lambda) = (1-\lambda)H\left(\sum_{i=1}^{\frac{1}{1-\lambda}}\delta_i\right)$.

In particular, in the low-rate region, the user's arrival pattern is completely leaked to the attacker:
\begin{equation}
\begin{aligned}
\underset{\lambda\to0}{\lim} \text{ } \mathcal{R}_{FCFS}(\lambda) =1.
\end{aligned}
\end{equation}
\label{theo:fcfs}
\end{theorem}
\begin{proof}
We first prove the  {\em converse} part  of~\eqref{eq:fcfs_main} by showing that there exists no attack strategy that allows the attacker to learn more information than~\eqref{eq:fcfs_main}.

From~\eqref{eq:privacy} and~\eqref{eq:queue_service}, we have
\begin{equation}
\begin{aligned}
\mathcal{L}_{FCFS} (\lambda)
&\overset{}{=}\underset{\mathbf{A}:\underset{k\to\infty}{\lim}\frac{k}{A_k}<1-\lambda}{\max}\text{ }
\underset{n\to\infty}{\lim}\frac{I\left(\boldsymbol{\delta}^{n};\mathbf{A}^m, q\left(A_1\right), q\left(A_2\right),\cdots, q\left(A_m\right)\right)}{n}\\
&\overset{(a)}{\leq} \underset{\mathbf{A}:\underset{k\to\infty}{\lim}\frac{k}{A_k}<1-\lambda}{\max}\text{ }
\underset{n\to\infty}{\lim} \frac{I\left(\boldsymbol{\delta}^{n}; \mathbf{A}^{m}, \mathbf{X}^{m}\right)}{n}\\
&\overset{}{=}\underset{\mathbf{A}:\underset{k\to\infty}{\lim}\frac{k}{A_k}<1-\lambda}{\max}\text{ }
\underset{n\to\infty}{\lim} \frac{H\left(\mathbf{X}^{m}\right)+H\left(\mathbf{A}^{m}|\mathbf{X}^{m}\right)-H\left(\mathbf{A}^{m}|\boldsymbol{\delta}^n\right)- H\left(\mathbf{X}^m|\mathbf{A}^m,\boldsymbol{\delta}^n\right)}{n}\\
&\overset{(b)}{=}\underset{\mathbf{A}:\underset{k\to\infty}{\lim}\frac{k}{A_k}<1-\lambda}{\max}\text{ }
\underset{n\to\infty}{\lim} \frac{H\left(\mathbf{X}^{m}\right)+H\left(\mathbf{A}^{m}|\mathbf{X}^{m}\right)-H\left(\mathbf{A}^{m}|\boldsymbol{\delta}^n\right)}{n}\\
&\overset{(c)}{\leq}\underset{\mathbf{A}:\underset{k\to\infty}{\lim}\frac{k}{A_k}<1-\lambda}{\max}\text{ }
\underset{n\to\infty}{\lim} \frac{H\left(\mathbf{X}^{m}\right)}{n}\\
\end{aligned}
\label{eq:temp_6}
\end{equation}
where $X_k$ is the number of the user's jobs that have arrived between time $A_{k-1}$ and $A_{k}$ and
$X_k=\sum_{j=A_{k-1}}^{A_{k}-1}\delta_{j}
\label{eq:pattern_1}$.
$(a)$ results from the application of  data processing inequality~\cite[Theorem 2.8.1]{Cover&Thomas:91} to the Markov chain
\begin{equation}
\begin{aligned}
\boldsymbol{\delta}^{n} \to {\mathbf{A}}^{{m}}, {\mathbf{X}}^{{m}} \to {\mathbf{A}}^{{m}}, q\left({A}_1\right),\cdots, q\left({A}_{{m}}\right).
\end{aligned}
\label{eq:info_flow}
\end{equation} 
$(b)$ follows from the fact that $\mathbf{X}^m$ is a deterministic function of $\mathbf{A}^m$ and $\boldsymbol{\delta}^n$, and $(c)$ results from the Markov chain\footnote{For FCFS, attack strategies can be divided into two types. The first type 
is fully independent with the user's behavior. 
The second type makes use of past departure history; $A_k$ depends on previous departure $D_{k-1}$. 
Since $D_{k-1}$ is uniquely determined once $\mathbf{X}^{k-1}$ is given,  ${\mathbf{A}}^{{m}}$ must be independent with
$\boldsymbol{\delta}^{n}$ given ${\mathbf{X}}^{{m}}$. This implies the Markov chain in~\eqref{eq:info_flow}.}
\begin{equation}
\boldsymbol{\delta}^{n} \to {\mathbf{X}}^{{m}} \to {\mathbf{A}}^{{m}}.
\label{eq:fcfs_attack_chain}
\end{equation} 

Recall the counting function in~\eqref{eq:count_func}. The number of user's jobs sent by time  $A_k$ is
$N(A_k)=\sum_{j=1}^{k}X_j$. Hence, \eqref{eq:temp_6} can be rewritten as
\begin{equation}
\begin{aligned}
&\mathcal{L}_{FCFS}(\lambda) \leq \underset{\mathbf{A}:\underset{k\to\infty}{\lim}\frac{k}{A_k}<1-\lambda}{\max}\text{ }
\underset{n\to\infty}{\lim} \frac{H\left(N\left(A_1\right),\cdots,N\left(A_m\right)\right)}{n}.
\end{aligned}
\label{eq:sampled_as_a_bound}
\end{equation}
This implies that the attacker learns at most a sampled version of the user's arrival process through this side channel.

From~\eqref{eq:m}, \eqref{eq:sampled_as_a_bound} can be rewritten as 
\begin{equation}
\begin{aligned}
\mathcal{L}_{FCFS}(\lambda)
&\leq \underset{\mathbf{A}:\underset{k\to\infty}{\lim}\frac{k}{A_k}<1-\lambda}{\max}\text{ }
\underset{n\to\infty}{\lim} \frac{H\left(N\left(A_1\right),\cdots,N\left(A_m\right)\right)}{m}\cdot \frac{m}{A_m} \cdot \frac{A_m}{n} \\
&= \underset{\mathbf{A}:\underset{k\to\infty}{\lim}\frac{k}{A_k}<1-\lambda}{\max}\text{ }\omega
\underset{n\to\infty}{\lim} \frac{H\left(N\left(A_1\right),\cdots,N\left(A_m\right)\right)}{m}. \end{aligned}
\end{equation}
Applying Theorem~\ref{theorem:periodic_sampling},  
\begin{equation}
\begin{aligned}
\mathcal{L}_{FCFS}(\lambda) &\overset{}{\leq}  \underset{\omega<1-\lambda}{\max}\text{ }\omega
\left (\alpha_\omega  H\left(\sum_{i=1}^{\left\lfloor \frac{1}{\omega} \right\rfloor}\delta_i\right)  +(1-\alpha_\omega) H\left(\sum_{i=1}^{\left\lceil \frac{1}{\omega} \right\rceil}\delta_i\right) \right) \\
&\overset{(d)}{=} (1-\lambda)\left (\alpha_{1-\lambda}  H\left(\sum_{i=1}^{\left\lfloor \frac{1}{1-\lambda} \right\rfloor}\delta_i\right)  +(1-\alpha_{1-\lambda}) H\left(\sum_{i=1}^{\left\lceil \frac{1}{1-\lambda} \right\rceil}\delta_i\right) \right),
\end{aligned}
\end{equation}
where $(d)$ simply applies the monotonic-increasing property of $H\left(\sum_{i=1}^{k}\delta_i\right)$ as a function of $k$.
This completes the proof for the converse.

To prove the achievability of~\eqref{eq:fcfs_main}, we consider the periodic sampling strategy defined in~\eqref{eq:attack_strategy_1}, and derive a lower bound on the information leakage, which turns out to meet the upper bound in~\eqref{eq:fcfs_main}. 
See Appendix~\ref{app:fcfs}.

Once~\eqref{eq:fcfs_main} is proven, we take the limit of leakage ratio as $\lambda \to 0$, 
\begin{equation}
\begin{aligned}
&\underset{\lambda\to0}{\lim} \text{ } \mathcal{R}_{FCFS}(\lambda)
 \overset{(e)}{=}  \frac{H\left(\delta_1\right)}{H(\lambda)} \overset{(f)}{=} 1.
\end{aligned}
\end{equation}
where $(e)$ holds because $\underset{\lambda\to0}{\lim}\alpha_{1-\lambda}=0$ according to~\eqref{eq:alpha} and $(f)$ follows from the Bernoulli distribution of $\delta_i$'s.
\end{proof}

 Theorem~\ref{theo:fcfs} proves that the attacker can recover the number of user's jobs arriving in each sampling period,  which becomes an accurate estimate of the user's job arrival pattern if the sampling frequency is  high. 
 When the user sends jobs at a very low rate, the attacker can sample  the queue state almost every time slot, and thus would learn the user's arrival pattern  completely (See Figure~\ref{fig:it_summary}).

\subsection{Round Robin}
\label{subsec:rr}
The next policy we study is round robin, where two users take turns to receive services. 
In each time slot, the service is switched to the next user who has jobs waiting in the queue; the scheduler never serves any single user continuously
unless the other user runs out of jobs. 
To derive a lower bound on the information leakage, we consider the nonstop monitoring attack introduced in~\S\ref{sec:lqf} (Figure~\ref{fig:nonstop_attack}), where 
the job arrival times and departure times satisfy~\eqref{eq:nonstop}. 
For round robin, this attack forces the scheduler to serve the attacker continuously if possible.
As a result,  the attacker 
learns when the user's queue becomes empty, as illustrated in  Figure~\ref{fig:rr-nonstop}, or 
\begin{equation}
\begin{aligned}
q\left({A}_k\right) \begin{cases}= 0 &  \quad \text{ if } {D}_k-{A}_k=1   \\ >0 &\quad \text{ if } {D}_k-{A}_k=2 \end{cases},
\quad k \in\mathbb{Z}.
\end{aligned}
\label{eq:rr-nonstop}
\end{equation}

Notice in~\eqref{eq:rr-nonstop}, the time gap between two consecutive  departures of attacker is at most two time slots. Hence, this attack is only applicable for the region, where the user's rate $\lambda\leq 0.5$.
\begin{figure}[t]
   \centering
   \includegraphics[width=0.6\columnwidth]{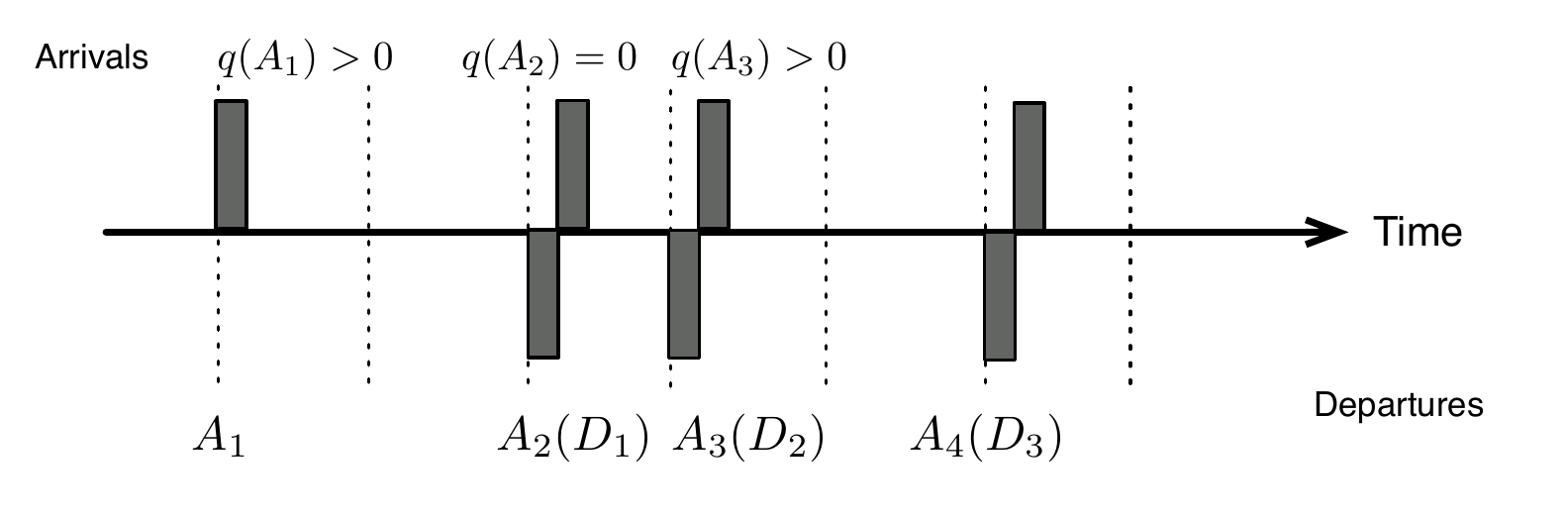} 
   \caption{{Nonstop monitoring on the round-robin scheduler:} the attacker's job are serviced instantly if there is no job from the user in the queue. Otherwise the scheduler needs to serve the user for one time slot before switching back to the attacker. As a result, the queuing delay ${D}_k-{A}_k$ indicates whether user's queue is empty, i.e., $q\left({A}_k\right)=0$.}
   \label{fig:rr-nonstop}
\end{figure}

Define the {\em busy period} of the system as the time gap between  two  times when the attacker finds the user's queue is empty, and denote
the $r^{th}$ busy period by $B_r$. $B_r$'s can be written as
\begin{equation}
\begin{aligned}
B_r = \inf\left\{{A}_k: {A}_k > \sum^{r-1}_{j=1} B_j \text{ and } q\left({A}_k\right)=0\right\} -\sum^{r-1}_{j=1} B_j.
\label{eq:rr_busy_def}
\end{aligned}
\end{equation}
What the attacker learns from the side channel is summarized by  the busy period sequence $\mathbf{B} =\{B_1, B_2, \cdots\}$.

\begin{theorem}
The information leakage of a round robin scheduler serving a user and an attacker is lower-bounded by
\begin{equation}
\begin{aligned}
\mathcal{L}_{RR}(\lambda)\geq \left(1-2\lambda\right) H(B), \qquad \text{for } \lambda\leq0.5,
\end{aligned}
\label{eq:rr_lb_main}
\end{equation}
 where $\lambda$ is the user's arrival rate, $B$ is the random variable distributed as
\begin{equation}
\begin{aligned}
\mathbb{P}\left(B=k \right) =\begin{cases} 2^{r-1}\lambda^{r-1}(1-\lambda)^{r}\max\left\{\sum_{j=1}^{\left\lfloor\frac{r-2}{2}\right\rfloor} \frac{(r-2)!}{(r-2-2j)!j!(j+1)!}2^{-2j-1},1\right\} &\text{ if }k=2r-1,\\
0 &  \text{ otherwise,}\end{cases}\label{eq:rr_busy_1}
\end{aligned}
\end{equation} 
where $r\in\mathbb{Z^{+}}$. 
In particular, when the user's rate is very low, the attacker learns completely the user's traffic pattern, i.e.,
\begin{equation}
\begin{aligned}
\underset{\lambda\to0}{\lim}\text{ } \mathcal{R}_{RR}(\lambda) =1. 
\label{eq:rr_ratio}
\end{aligned}
\end{equation}
\label{theo:rr_lb}
\end{theorem}
\begin{proof}
Similar to~\eqref{eq:lqf_lb},  the nonstop monitoring attack results in a lower bound on the leakage given by
\begin{equation}
\begin{aligned}
\mathcal{L}_{RR}(\lambda) &\overset{}{\geq} \underset{n\to\infty}{\lim}\frac{I\left(\boldsymbol{\delta}^{n}; {\mathbf{D}}^{{m}}\right)}{n}\\
& \overset{(a)}{=}\underset{n\to\infty}{\lim}\frac{H\left({\mathbf{D}}^{{m}}\right)}{n}
\end{aligned}
\label{eq:temp_8_0}
\end{equation}
where $(a)$ holds because the attacker's departure times for the round robin scheduler are deterministic once the user's arrival pattern is known. 
From~\eqref{eq:nonstop} and \eqref{eq:rr-nonstop}, we have
\begin{equation}
\begin{aligned}
H\left({\mathbf{D}}^{{m}}\right) &\overset{}{=}H\left(\mathbbm{1}_{\left\{q\left({A}_1\right)=0\right\}}, \mathbbm{1}_{\left\{q\left({A}_2\right)=0\right\}}, \cdots, \mathbbm{1}_{\left\{q\left({A}_{{m}}\right)=0\right\}}\right)\\
&\overset{(b)}{\geq}H\left(B_1,B_2, \cdots,B_{{m}}\right)
\end{aligned}
\label{eq:temp_8}
\end{equation}
where $(b)$ follows from the definition of busy periods in~\eqref{eq:rr_busy_def}. 

It can be further shown that the busy periods seen by the attacker, $B_r$'s,  are i.i.d. distributed as~\eqref{eq:rr_busy_1} and have mean $\frac{1}{1-2\lambda}$.
The proof is presented in Appendix~\ref{app:rr_busy}. 
Therefore, 
plugging~\eqref{eq:temp_8} into~\eqref{eq:temp_8_0} proves~\eqref{eq:rr_lb_main}. 
Additionally, taking the limit of~\eqref{eq:rr_lb_main} at $\lambda\to0$,  the  leakage ratio is lower-bounded as 
\begin{equation}
\begin{aligned}
\underset{\lambda\to0}{\lim}\text{ } \mathcal{R}_{RR} (\lambda)&\geq  \underset{\lambda\to0}{\lim} (1-2\lambda)\frac{H\left(B\right)}{H(\lambda)}\\
&\overset{(c)}{\geq} \underset{\lambda\to0}{\lim} \frac{-(1-\lambda)\log(1-\lambda)-\lambda(1-\lambda)^2\log(\lambda(1-\lambda)^2)}{-(1-\lambda)\log(1-\lambda)-\lambda\log\lambda},
\end{aligned}
\label{eq:temp_81}
\end{equation}
where  $(c)$ holds by plugging in  the PMF of random variable $B$ in~\eqref{eq:rr_busy_1} only for the terms $k=1$ and $k=3$. 
The limit on the right hand side of inequality $(c)$ goes to 1 as $
\lambda\to0$. This completes the proof. 
\end{proof}

Round robin is one of the simplest scheduling algorithms in multi-processor operation systems and known for fairness~\cite{hahne86}. 
However, our analysis shows that in the low-rate region the round robin scheduler almost entirely leaks a user's traffic pattern through the timing side channel (See Figure~\ref{fig:it_summary}).

\subsection{Work-Conserving TDMA}
\label{subsec:wctdma}
The schedulers analyzed so far all leak substantial information about the user's traffic, especially when the user's job arrival rate is low. 
 This imposes serious
threat to user privacy since many network systems have light workloads. 
In fact, studies have shown that average server utilization in real world data centers is only about 5\% to 20\%~\cite{siegele2008let}. 
In this section, we study WC-TDMA, a tweak of TDMA,
which can reduce the information leakage in the low-rate traffic region by half.

In WC-TDMA, time slots are pre-assigned to the user and attacker equally. Unlike TDMA,
if in some slot the reserved user has no jobs left, the WC-TDMA scheduler serves the other
user who has jobs waiting in the queue. For the sake of convenience, we assume all odd
time slots are assigned to the user, and all even slots are assigned to the attacker.

We consider an attack strategy, where the attacker sends a job in each odd slot and stays idle in all even slots, i.e., 
\begin{equation}
\begin{aligned}
{A}_k = 2k-1, \quad k \in\mathbb{Z}.
\label{eq:half_attack}
\end{aligned}
\end{equation}
Clearly, this attack consumes half of the service capacity, so is only applicable when the user's rate $\lambda\leq 0.5$.

Since odd time slots are reserved for the user, the attacker's jobs sent on those slots would experience delays as follows:\begin{equation}
\begin{aligned}
{D}_k-{A}_k =\begin{cases} 1 & \text{ if } q\left({A}_k\right)=0 \text{ and }\delta_{{A}_k} =0
 \\ 2 & \text{ otherwise} \end{cases}, \quad k\in\mathbb{Z}.
\label{eq:nitdma_q}
\end{aligned}
\end{equation} 
Therefore, the attacker learns  the time slots when the queue is empty and the user has not issued a job. 

Define the {`busy period'}, $B'_r, r\in\mathbb{Z}$, to be the time gap between two successive times when the attacker sees an empty queue. Then,
\begin{equation}
\begin{aligned}
B'_r = \inf\left\{{A}_k: {A}_k> \sum^{r-1}_{j=1} B'_j \text{ and } q\left({A}_k\right)=0\right\} -\sum^{r-1}_{j=1} B'_j.
\label{eq:nitdma_busy_def}
\end{aligned}
\end{equation}

\begin{theorem}
The information leakage ratio of a WC-TDMA scheduler serving a user and an attacker  is lower-bounded  by
\begin{equation}
\begin{aligned}
\mathcal{L}_{WC-TDMA}(\lambda)\geq  \frac{1-2\lambda}{2-2\lambda}H\left(B\right), \qquad \text{for } \lambda\leq0.5,
\label{eq:nitdma_lb}
\end{aligned}
\end{equation}
 where $\lambda$ is the user's arrival rate, and $B$ is the random variable distributed as~\eqref{eq:rr_busy_1}.
\label{theo:nitdma}
\end{theorem}
\begin{proof}
Like~\eqref{eq:temp_8_0} in the proof of Thereom~\ref{theo:rr_lb}, a lower bound on information leakage can be written as 
\begin{equation}
\begin{aligned}
\mathcal{L}_{WC-TDMA}(\lambda)
&\overset{}{\geq}\underset{k\to\infty}{\lim}\frac{H\left({\mathbf{D}}^{k}\right)}{2k}\\
&\overset{(a)}{=}\underset{k\to\infty}{\lim}\frac{H\left(\mathbbm{1}_{\left\{q\left(1\right)=0\right\}}, \mathbbm{1}_{\left\{q\left(3\right)=0\right\}}, \cdots, \mathbbm{1}_{\left\{q\left(2k-1\right)=0\right\}}\right)}{2k}\\
&\overset{(b)}{\geq}\underset{k\to\infty}{\lim}\frac{H\left(B'_1,B'_2, \cdots, B'_{{k}}\right)}{2k}\\
\end{aligned}
\label{eq:nitdma_lb}
\end{equation}
where $(a)$ follows from~\eqref{eq:half_attack} and~\eqref{eq:nitdma_q},  and $(b)$ follows from~\eqref{eq:nitdma_busy_def}.

It can be shown that the busy periods seen by the attacker, $B'_r$'s,  are i.i.d. distributed as~\eqref{eq:rr_busy_1} and have mean $\frac{2-2\lambda}{1-2\lambda}$ (See  Appendix~\ref{app:nitdma_busy}). Hence, \eqref{eq:nitdma_lb} readily implies the desired lower bound.   
\end{proof}
\begin{corollary}
The information leakage of a WC-TDMA scheduler 
serving a user and an attacker is given by 
\begin{equation}
\begin{aligned}
\underset{\lambda\to0}{\lim}\text{ }  \mathcal{R}_{WC-TDMA}(\lambda) =\frac{1}{2}, \qquad \text{if } \lambda\to0.
\end{aligned}
\label{eq:wctdma_cor}
\end{equation}
\label{cor:nitdma}
\end{corollary}

\begin{proof}
Taking the  limit of~\eqref{eq:nitdma_lb} at $\lambda\to0$,  we get
\begin{equation}
\begin{aligned}
\underset{\lambda\to0}{\lim}\text{ } \mathcal{R}_{WC-TDMA}(\lambda)&\geq  \underset{\lambda\to0}{\lim} \frac{1-2\lambda}{2-2\lambda}  \frac{H\left(B\right)}{H(\lambda)}\
\overset{(a)}{\geq}\frac{1}{2},
\end{aligned}
\label{eq:wctdma_r_lb}
\end{equation}
where $(a)$ follows from~\eqref{eq:temp_81}.

We subsequently derive an upper-bound on the leakage ratio. 
Define $S_k$ to be the earliest time when the $k^{th}$ attacker's job can receive service. Then  $S_k=\max\{D_{k-1}, A_k\}$, and we have 
 \begin{equation}
\begin{aligned}
I\left(\boldsymbol{\delta}^{n}; \mathbf{A}^{m},\mathbf{D}^{m}\right)&=I\left(\boldsymbol{\delta}^{n}; \mathbf{S}^{m},\mathbf{D}^{m}\right)\\
&\leq H\left(\mathbf{S}^{m},\mathbf{D}^{m}\right)\\
&=H\left(D_1-S_1, \cdots, D_m-S_m \right)\\
&\overset{(b)}{\leq} \underset{k:S_k \text{ is odd}}{\sum}H\left(D_k-S_k\right)
 \end{aligned}
  \label{eq:temp_11}
\end{equation}
where $(b)$ follows from the fact that 
the attacker is served with priority in even time slots, i.e., 
 $D_k - S_k \equiv 1$, where  $S_k$ is even.
 On the other hand, since the user is served first during the odd time slots, 
similar to~\eqref{eq:nitdma_q} we have
\begin{equation}
\begin{aligned}
D_k- S_k=\begin{cases} 1  \quad \text{if }q(S_k) =0  \text{ and }\delta_{S_k} =0, \\2 \quad \text{otherwise.} \end{cases}
\label{eq:temp_10}
\end{aligned}
\end{equation} 

Applying~\eqref{eq:temp_10} to~\eqref{eq:temp_11} and taking the limit  at $\lambda \to 0$, we have
\begin{equation}
\begin{aligned}
\underset{\lambda\to0}{\lim}\text{ } I\left(\boldsymbol{\delta}^{n}; \mathbf{A}^{m},\mathbf{D}^{m}\right)&
\leq  \underset{k:S_k \text{ is odd}}{\sum} \underset{\lambda\to0}{\lim} \text{ }
H\left(\mathbb{P}\left(q(S_k) =0\right)\cdot(1-\lambda)\right)\\
&\overset{(c)}{=}   \underset{k:S_k \text{ is odd}}{\sum} H\left(\lambda\right),
\end{aligned}
\label{eq:temp_12}
\end{equation}
where $(c)$ holds because
\begin{equation}
\begin{aligned}
\underset{\lambda\to0}{\lim}   \mathbb{P}\left(q(S_k) =0\right) = 1, \quad k \in\mathbb{Z}.
\end{aligned}
\label{eq:temp_13}
\end{equation}
See Appendix~\ref{app:eq_temp_13} for the proof of~\eqref{eq:temp_13}.

Plugging~\eqref{eq:temp_13} into~\eqref{eq:temp_12},  we have 
 \begin{equation}
\begin{aligned}
\underset{\lambda\to0}{\lim}\text{ } \mathcal{R}_{WC-TDMA}(\lambda) &\leq \underset{\mathbf{S}:\underset{k\to\infty}{\lim}\frac{k}{{S}_k}<1-\lambda}{\max}\quad\underset{n\to\infty}{\lim}\frac{\underset{k:S_k \text{ is odd}}{\sum}H\left(D_k-S_k\right)}{nH(\lambda)}\\
 &\overset{(d)}{\leq}  \frac{1}{2}.
\end{aligned}
\label{eq:temp_14}
\end{equation}
 where $(d)$ holds because at most half of time slots are odd. 
\eqref{eq:temp_14}  together with \eqref{eq:wctdma_r_lb} prove~\eqref{eq:wctdma_cor}. 
\end{proof}

\subsection{Deterministic Work-Conserving Schedulers}

 In this section, we derive a universal lower bound on the information leakage for the class of deterministic work-conserving (det-WC) schedulers, where the scheduler's actions are deterministic (the same arrival instances result in the same departure events), and the scheduler can idle only if there are no jobs in the queue.

\begin{theorem}
The information leakage of 
 a det-WC scheduler serving  a user and an attacker is lower-bounded by
\begin{equation}
\begin{aligned}
\mathcal{L}_{det-WC}(\lambda) \geq  \underset{\omega:\omega<1-\lambda}{\max} \frac{\omega(z_0-1)}{2z_0} H(\lambda),
\end{aligned}
\label{eq:det-wc_main}
\end{equation}
where  $\lambda$ is the user's arrival rate, and $z_0$ is the only root of the equation
\begin{equation}
\begin{aligned}
\alpha_{\omega}z(1-\lambda+\lambda z)^{\left\lceil\frac{1}{\omega}\right\rceil} +
(1-\alpha_{\omega})z^2(1-\lambda+\lambda z)^{\left\lfloor\frac{1}{\omega}\right\rfloor} -z^{\left\lceil\frac{1}{\omega}\right\rceil} =0, 
\end{aligned}
\label{eq:z_0}
\end{equation} 
outside the unit circle.
\label{theo:detwc}
\end{theorem}
\begin{proof}
We again consider the periodic-sampling attack, $\mathbf{{A}}^*$, defined in~\eqref{eq:attack_strategy_1}, under which  the mutual information between the attacker's observation and the user's arrival pattern is given by
\begin{equation}
\begin{aligned}
I\left(\boldsymbol{\delta}^{n}; \mathbf{A}^{*m},\mathbf{D}^{*m}\right)
&\overset{(a)}{=} H\left(\boldsymbol{\delta}^{n}\right)- \sum_{i=1}^{n}H\left(\delta_i\big|\boldsymbol{\delta}^{i-1},{\mathbf{A}}^{{*m}},{\mathbf{D}}^{{*m}}\right)\\
&\overset{(b)}{=}H\left(\boldsymbol{\delta}^{n}\right)- \sum_{i=1}^{n}H\left(\delta_i\big|q\left(1\right), \cdots, q\left(i\right), {\mathbf{A}}^{{*m}},{\mathbf{D}}^{{*m}}\right)\\
&\geq H\left(\boldsymbol{\delta}^{n}\right)- \sum_{i=1}^{n}H\left(\delta_i\big|q\left(i\right), {\mathbf{A}}^{{*m}}\right)
\end{aligned}
\label{eq:temp_15_2}
\end{equation}
where $(a)$ follows from the entropy chain rule. 
$(b)$ holds because 
given $\boldsymbol{\delta}^{i-1}$ and ${\mathbf{A}}^{{*m}}$,  queue lengths at all time slots up to time $i$ are known. 

Define $a_i$ as the indicator of the attacker's arrival event in time slot $i$. Following~\eqref{eq:attack_strategy_1}, $a_i$'s are i.i.d. distributed as 
\begin{equation}
\begin{aligned}
a_i=\begin{cases}1 \quad \text{w.p. } \omega \\
0 \quad \text{w.p. } 1-\omega
\end{cases}, \quad i\in\mathbb{Z}.
\end{aligned}
\label{eq:temp_15}
\end{equation}
Applying~\eqref{eq:temp_15} to~\eqref{eq:temp_15_2}, we  have
\begin{equation}
\begin{aligned}
I\left(\boldsymbol{\delta}^{n}; \mathbf{A^*}^{m},\mathbf{D^*}^{m}\right)
&\geq H\left(\boldsymbol{\delta}^{n}\right)- \sum_{i=1}^{n}H\left(\delta_i\big|q\left(i\right), a_i\right)
\end{aligned}
\label{eq:temp_15_1}
\end{equation}

At some time slot $i$, assume the scheduler serves the user first when a tie of queue length occurs.
Additionally, assume the queue is empty, i.e., $q\left(i\right)=0$, and the attacker issues  a new job, i.e., $a_i=1$. 
As argued before, in this case the attacker knows whether the user issued a job or not in this slot without ambiguity. 
Since the scheduler does not know the users' identities, in the worst case, the user has  priority  for at least half of time slots.
Therefore, the afore mentioned scenario would at least occur with probability by $\frac{1}{2}\omega\mathbb{P}\left(q(i)=0\right)$.  Plugging this  into~\eqref{eq:temp_15_1}, we have
\begin{equation}
\begin{aligned}
I\left(\boldsymbol{\delta}^{n}; \mathbf{A^*}^{m},\mathbf{D^*}^{m}\right)
&\geq H\left(\boldsymbol{\delta}^{n}\right)- \sum_{i=1}^{n}  \left(1-\frac{1}{2}\omega  \mathbb{P}\left(q(i)=0\right)\right)H\left(\delta_i\right)\\
&= \frac{1}{2} \omega  \mathbb{P}\left(q(i)=0\right)H\left(\lambda\right).
\end{aligned}
\label{eq:temp_16}
\end{equation}
We next derive the probability that the attacker sees an empty queue, $ \mathbb{P}\left(q(i)=0\right)$.
From~\eqref{eq:attack_strategy_1}, we can write the update equation for queue lengths sampled by the attacker as
\begin{equation}
\begin{aligned}
q\left({A}^*_{k+1}\right)=\left(q\left({A}^*_{k+1}\right) + Y_k - \left\lceil \frac{1}{\omega} \right\rceil \right)_+, \quad k\in\mathbb{Z},
\end{aligned}
\label{eq:temp_18}
\end{equation}
where $Y_k$'s are i.i.d. distributed as 
\begin{equation}
\begin{aligned}
Y_k=\begin{cases} 2+X^*_k  & \mbox{w.p. } \alpha_\omega \\ 1+{X^*}'_k & \mbox{w.p. } 1-\alpha_\omega \end{cases},
\end{aligned}
\end{equation}
where $X^*_k\sim Binomial\left(\left\lfloor \frac{1}{\omega} \right\rfloor,\lambda\right)$, ${X^*}'_k\sim Binomial\left(\left\lceil \frac{1}{\omega} \right\rceil ,\lambda\right)$. 
The probability of the attacker seeing an empty queue can be derived by calculating $z$-transform of $q\left({A}_{k+1}\right)$ in the steady state~\cite[(3)]{bruneel86}, which is given by
\begin{equation}
\begin{aligned}
\underset{k\to\infty}{\lim}\mathbb{P}\left(q\left({A}^*_{k+1}\right)=0\right) = \frac{z_0-1}{z_0},
\end{aligned}
\label{eq:temp_19}
\end{equation}
where $z_0$ is the only root of $\alpha_{\omega}z(1-\lambda+\lambda z)^{\left\lceil\frac{1}{\omega}\right\rceil} +
(1-\alpha_{\omega})z^2(1-\lambda+\lambda z)^{\left\lfloor\frac{1}{\omega}\right\rfloor} -z^{\left\lceil\frac{1}{\omega}\right\rceil} =0$ outside  the unit circle. 

\eqref{eq:temp_16} and~\eqref{eq:temp_19} readily imply~\eqref{eq:det-wc_main}.
\end{proof}

We plot the numerical solution of the bound in~\eqref{eq:det-wc_main} in Figure~\ref{fig:it_summary}. 
As can be seen,  the attacker always gains a significant amount of information of the user in a det-WC scheduler, especially when the user's rate $\lambda$ is below $0.5$. 
More specifically, as $\lambda\to 0$,
  the user's traffic pattern is leaked for at least half of the time.

\section{Conclusion}
Timing side channels in deterministic work-conserving schedulers were studied, where information leakage happens due to the sharing of 
queue among users. 
Our analysis proved that commonly deployed work-conserving queue service policies all  leak significant amounts of user's information. 
If the scheduler  adds idling slots, or randomizes the service oder of jobs to some extent, the attacker may not be able to make accurate inferences about queue states any more and subsequently the user's arrival pattern.
Such mitigation measures for the timing side channel of our interest remain open. 

\newpage
\appendix

\subsection{Proof of Theorem~\ref{theorem:periodic_sampling}}
\label{app:opt_samp}

%
Before proving Theorem~\ref{theorem:periodic_sampling}, we introduce two lemmas. 

Define function $\mathcal{H}_{\lambda}$: $\mathbb{Z^{+}} \to \mathbb{R}\cup\{+\infty\}$ as 
\begin{equation}
\begin{aligned}
\mathcal{H}_{\lambda}(i)= H\left(\delta_1,\delta_2,\cdots,\delta_i\Big|\sum_{j=1}^{i}\delta_j\right),
\end{aligned}
\label{eq:def_f}
\end{equation}
where $\delta_i$'s are i.i.d. Bernoulli random variables.

\begin{lemma} $\mathcal{H}_{\lambda}(\cdot)$ is a mid-point convex function~\cite[(2.8)]{Murota}; i.e., 
\begin{equation}
\begin{aligned}
\mathcal{H}_{\lambda}(a)+\mathcal{H}_{\lambda}(b)\geq \mathcal{H}_{\lambda}\left(\left\lfloor \frac{a+b}{2} \right\rfloor\right) +\mathcal{H}_{\lambda}\left(\left\lceil \frac{a+b}{2} \right\rceil\right), 
\end{aligned}
\end{equation}
for all $a\leq b, \;   a, b \in\mathbb{Z^{+}}$, and the equality is achieved if $b = a \text{ or } b=a+1$. 
\label{lem:mid}
\end{lemma}
\begin{proof}
Given $m, n\in \mathbb{Z^{+}}$, $m\geq n$, we compute 
\begin{equation}
\begin{aligned}
\mathcal{H}_{\lambda}(m)-\mathcal{H}_{\lambda}\left(n\right)=&H\left(\boldsymbol{\delta}^{n}, \boldsymbol{\delta}_{n+1}^{m}\Big|\sum_{i=1}^{m}\delta_i\right)-H\left(\boldsymbol{\delta}^{n}\Big|\sum_{i=1}^{n}\delta_i\right)\\
\overset{(a)}{=}&H\left( \boldsymbol{\delta}_{n+1}^{m}\Big|\sum_{i=1}^{m}\delta_i\right)+H\left(\boldsymbol{\delta}^{n}\Big| \boldsymbol{\delta}_{n+1}^{m},\sum_{i=1}^{m}\delta_i\right)-H\left(\boldsymbol{\delta}^{n}\Big|\sum_{i=1}^{n}\delta_i\right)\\
\overset{(b)}{=}&H\left(\boldsymbol{\delta}_{n+1}^{m}\Big|\sum_{i=1}^{m}\delta_i\right)\\
\overset{(c)}{=}&H\left(\boldsymbol{\delta}^{m-n}\Big|\sum_{i=1}^{m}\delta_i\right)
\end{aligned}
\label{eq:1}
\end{equation}
where $(a)$ applies the entropy chain rule~\cite[Theorem~2.2.1]{Cover&Thomas:91}, $(b)$ holds because 
that  $\sum_{i=1}^{n}\delta_i$ is a sufficient statistic to infer $\boldsymbol{\delta}^n$ and
can be calculated from $\boldsymbol{\delta}_{n+1}^m$ and $\sum_{i=1}^{m}\delta_i$,  and $(c)$ follows from the fact that Bernoulli arrivals are uniformly distributed once the total number of arrivals is known. 

Replacing $(m, n)$ with $\left(b,\left\lceil \frac{a+b}{2} \right\rceil\right)$ and $\left(\left\lfloor\frac{a+b}{2}\right\rfloor, a\right)$  in~\eqref{eq:1}, respectively, we get 
\begin{equation}
\begin{aligned}
\mathcal{H}_{\lambda}(b)-\mathcal{H}_{\lambda}\left(\left\lceil \frac{a+b}{2} \right\rceil\right)
=H\left(\boldsymbol{\delta}^{\left\lfloor\frac{a+b}{2}\right\rfloor-a}\Big|\sum_{i=1}^{b}\delta_i\right)
\label{eq:2}
\end{aligned}
\end{equation}
and
\begin{equation}
\begin{aligned}
\mathcal{H}_{\lambda}\left(\left\lfloor\frac{a+b}{2}\right\rfloor\right)-\mathcal{H}_{\lambda}\left(a\right)
=H\left(\boldsymbol{\delta}^{\left\lfloor\frac{a+b}{2}\right\rfloor-a}\Big|\sum_{i=1}^{\left\lfloor\frac{a+b}{2}\right\rfloor}\delta_i\right).
\label{eq:3}
\end{aligned}
\end{equation}

Subtracting~\eqref{eq:3} from~\eqref{eq:2}, we get
\begin{equation}
\begin{aligned}
\mathcal{H}_{\lambda}(b)+\mathcal{H}_{\lambda}\left(a\right)
-\mathcal{H}_{\lambda}\left(\left\lceil \frac{a+b}{2} \right\rceil\right)
-\mathcal{H}_{\lambda}\left(\left\lfloor\frac{a+b}{2}\right\rfloor\right)
=&H\left(\boldsymbol{\delta}^{\left\lfloor\frac{a+b}{2}\right\rfloor-a}\Big|\sum_{i=1}^{b}\delta_i\right)
-H\left(\boldsymbol{\delta}^{\left\lfloor\frac{a+b}{2}\right\rfloor-a}\Big|\sum_{i=1}^{\left\lfloor\frac{a+b}{2}\right\rfloor}\delta_i\right)\\
\overset{(d)}{=}&H\left(\boldsymbol{\delta}^{\left\lfloor\frac{a+b}{2}\right\rfloor-a}\Big|\sum_{i=1}^{b}\delta_i\right)-\
H\left(\boldsymbol{\delta}^{\left\lfloor\frac{a+b}{2}\right\rfloor-a}\Big|\sum_{i=1}^{\left\lfloor\frac{a+b}{2}\right\rfloor}\delta_i,\sum^{b}_{i=1}\delta_i\right)\\
= & I\left(\boldsymbol{\delta}^{\left\lfloor\frac{a+b}{2}\right\rfloor-a}; \sum_{i=1}^{\left\lfloor\frac{a+b}{2}\right\rfloor}\delta_i\Big|\sum_{i=1}^{b}\delta_i\right)
\geq 0
\end{aligned}
\label{eq:l3}
\end{equation}
where $(d)$ follows from that fact that $\sum_{i=\left\lfloor\frac{a+b}{2}\right\rfloor+1}^{b}\delta_i$  does not provide extra information to infer $\boldsymbol{\delta}^{\left\lfloor\frac{a+b}{2}\right\rfloor-a}$.
Clearly, the last inequality turns equality if $b=a \text{ or }a+1$, which completes the proof. 
\end{proof}

\begin{lemma}
$\mathcal{H}_{\lambda}(\cdot)$ is an integrally-convex function~\cite[(2.5)]{Murota}; i.e., it can be extended to a globally convex function $\hat{\mathcal{H}}_{\lambda}$: $\mathbb{R}^{+}\to \mathbb{R}\cup\{+\infty\}$, where
\begin{equation}
\begin{aligned}
\hat{\mathcal{H}}_{\lambda}(x)=\alpha_x \mathcal{H}_{\lambda}( \left\lfloor x \right\rfloor)  +(1-\alpha_x) \mathcal{H}_{
\lambda}(\left\lceil x \right\rceil)
\label{eq:def_f_hat}
\end{aligned}
\end{equation}
where $\alpha_x=\frac{\left\lceil \frac{1}{x} \right\rceil-\frac{1}{x}}{\left\lceil \frac{1}{x} \right\rceil-\left\lfloor \frac{1}{x} \right\rfloor}$. 
\label{lem:integral}
\end{lemma}
\begin{proof}
Applying Lemma~\ref{lem:mid}, this is a direct result from~\cite[Theorem~2.4]{Murota}, which states that a discrete function satisfying mid-point convexity can be extended to a convex continuous function through linear interpolation. 
\end{proof}
%

{\bf Proof of Theorem~\ref{theorem:periodic_sampling}:}
The entropy of the sampled process satisfies
\begin{equation}
\begin{aligned}
H\left(N\left(A_1\right),N\left(A_2\right),\cdots,N\left(A_k\right)\right) &\overset{(a)}{=} I\left(\boldsymbol{\delta}^{A_k-1};N\left(A_1\right),N\left(A_2\right),\cdots,N\left(A_k\right)\right)\\
=&I\left(\boldsymbol{\delta}^{A_k-1};N\left(A_2\right)-N\left(A_1\right),\cdots,N\left(A_k\right)-N\left(A_{k-1}\right)\right)\\
 \overset{(b)}{=}& H\left(\boldsymbol{\delta}^{A_k-1}\right) - \sum_{i=1}^{k-1}H\left(\boldsymbol{\delta}^{A_{i+1}-1}_{A_i}\big|N\left(A_{i+1}\right)-N\left(A_{i}\right)\right),
 \end{aligned}
\label{eq:app_temp_10}
\end{equation}
where $(a)$ follows from that  $N(A_i)$'s are functions of $\boldsymbol{\delta}$ and $(b)$ applies the entropy chain rule.

Define $n_r$ to be the number of elements in the sequence $\{A_2-A_1, \cdots, A_k-A_{k-1}\}$ that take value of $r$, then $\sum_{r=1}^{\infty} n_r=k$.
The conditional entropy in~\eqref{eq:app_temp_10} can be rewritten as   
\begin{equation}
\begin{aligned}
\sum_{i=1}^{k-1}H\left(\boldsymbol{\delta}^{A_{i+1}-1}_{A_i}\big|N\left(A_{i+1}\right)-N\left(A_{i}\right)\right)&=  \sum_{r=1}^{\infty}n_rH\left(\boldsymbol{\delta}^{r}\big|\sum_{j=1}^{r}\delta_j\right)\\
&=\sum_{r=1}^{\infty}n_r{\mathcal{H}}_\lambda(r)\\
&\overset{(c)}{=}\sum_{r=1}^{\infty}n_r\hat{\mathcal{H}}_\lambda(r),
\end{aligned}
\label{eq:app_temp_11}
\end{equation}
where function $\mathcal{H}_{\lambda}(\cdot)$ and $\hat{\mathcal{H}}_{\lambda}(\cdot)$ are defined according to~\eqref{eq:def_f} and~\eqref{eq:def_f_hat}, respectively, and $(c)$ holds because $\mathcal{H}_{\lambda}(\cdot)$ and $\hat{\mathcal{H}}_{\lambda}(\cdot)$ take the same values for integer arguments. 

Combining~\eqref{eq:app_temp_10} and~\eqref{eq:app_temp_11}, the entropy rate of the sampled process is upper-bounded as given by
\begin{equation}
\begin{aligned}
\underset{k\to\infty}{\lim}\frac{H\left(N\left(A_1\right),\cdots,N\left(A_k\right)\right)}{k}&=\omega H(\delta_1)-\underset{k\to\infty}{\lim}\frac{\sum_{r=1}^{\infty}n_r\hat{\mathcal{H}}_\lambda(r)}{k}\\
&\overset{(d)}{\leq} \omega H(\delta_1)-\underset{k\to\infty}{\lim}\hat{\mathcal{H}}_\lambda\left(\frac{\sum_{r=1}^{\infty}n_rr}{k}\right)\\
&\leq\omega H(\delta_1)-\hat{\mathcal{H}}_\lambda\left(\underset{k\to\infty}{\lim}\frac{\sum_{r=1}^{\infty}n_rr}{k}\right)\\
&\overset{(e)}= \omega H(\delta_1)-\hat{\mathcal{H}}_\lambda\left(\frac{1}{\omega}\right)\\
&\overset{(f)}=\alpha_\omega  H\left(\sum_{i=1}^{\left\lfloor \frac{1}{\omega} \right\rfloor}\delta_i\right)  +(1-\alpha_\omega) H\left(\sum_{i=1}^{\left\lceil \frac{1}{\omega} \right\rceil}\delta_i\right),
\end{aligned}
\label{eq:app_temp_12}
\end{equation}
where $(d)$ applies Jensen's inequality~\cite[(9.1.3.1)]{krantz95}, $(e)$ follows from $\omega =\underset{k\to\infty}{\lim}\frac{k}{{A}_k}$, and $(f)$ follows from~\eqref{eq:def_f} and~\eqref{eq:def_f_hat}. 

Finally, it is easy to verify that the sampled process by the periodic-sampling strategy has the same entropy rate as given by the bound in~\eqref{eq:app_temp_12}, 
which completes the proof for~\eqref{eq:periodic_sampling_main}.

\subsection{Proof of Achievability in Theorem~\ref{theo:fcfs}}
\label{app:fcfs}
For the achievability of~\eqref{eq:fcfs_main}, we need to show
\begin{equation}
\begin{aligned}
&\mathcal{L}_{FCFS}(\lambda)\geq\nonumber(1-\lambda)\left( \alpha_{1-\lambda}
{H\left(\sum_{i=1}^{\left\lfloor \frac{1}{1-\lambda} \right\rfloor}\delta_i\right)} + (1-\alpha_{1-\lambda}){H\left(\sum_{i=1}^{\left\lceil \frac{1}{1-\lambda} \right\rceil}\delta_i\right)}\right).
\end{aligned}
\label{eq:achievability}
\end{equation}
\begin{proof}
Let $\mathbf{D}^*=\{D^*_1,D^*_2,\cdots\}$ denote the departure times of the attacker's jobs applying
the {periodic sampling} attack strategy defined in~\eqref{eq:attack_strategy_1}. We have a lower bound on the information leakage as 
\begin{equation}
\begin{aligned}
\mathcal{L}_{FCFS} (\lambda)
&\geq \underset{
\substack{\mathbf{A}^*:\underset{k\to\infty}{\lim}\frac{k}{{A}^*_k}<1-\lambda}}
{\max}\text{ }
\underset{n\to\infty}{\lim}\frac{I\left(\boldsymbol{\delta}^{n}; {\mathbf{A}}^{{*m}},{\mathbf{D}}^{{*m}}\right)}{n}.
\end{aligned}
\label{eq:fcfs_lb}
\end{equation}
Rewrite the mutual information in~\eqref{eq:fcfs_lb} as follows:
\begin{equation}
\begin{aligned}
I\left(\boldsymbol{\delta}^{n}; {\mathbf{A}}^{{*m}},{\mathbf{D}}^{{*m}}\right)\overset{(a)}{=}& I\left(\boldsymbol{\delta}^{n};{\mathbf{A}}^{{*m}}, q\left({A}^*_1\right), q\left({A}^*_2\right),\cdots, q\left({A}^*_{{m}}\right)\right)\\
\overset{(b)}{=}& I\left(\boldsymbol{\delta}^{n}; {\mathbf{A}}^{{*m}}, {\mathbf{X}}^{{*m}}\right)-I\left(\boldsymbol{\delta}^{n}; {\mathbf{A}}^{{*m}}, {\mathbf{X}}^{{*m}}\big|{\mathbf{A}}^{{*m}}, q\left({A}^*_1\right), \cdots, q\left({A}^*_{{m}}\right)\right)\\
=&I\left(\boldsymbol{\delta}^{n}; {\mathbf{A}}^{{*m}}, {\mathbf{X}}^{{*m}}\right)
-H\left({\mathbf{X}}^{{*m}}\big|{\mathbf{A}}^{{*m}}, q\left({A}^*_1\right),\cdots, q\left({A}^*_{{m}}\right)\right)\\
\overset{(c)}{=}&H\left({\mathbf{X}}^{{{*m}}}\big| {\mathbf{A}}^{{{*m}}}\right)-H\left({\mathbf{X}}^{{{*m}}}\big|{\mathbf{A}}^{{*m}}, q\left({A}^*_1\right),\cdots, q\left({A}^*_{{m}}\right)\right),
\end{aligned}
\label{eq:temp_0}
\end{equation}
where $(a)$ follows from~\eqref{eq:queue_service}, ${X}^*_k$ is the total number of the user's jobs that have arrived between the times ${A}^*_{k-1}$ and ${A}^*_{k}$, 
$(b)$ follows from the Markov chain in~\eqref{eq:info_flow}, and $(c)$ holds because ${\mathbf{A}}^{{*m}}$ is independent of $\boldsymbol{\delta}^{n}$. 

Substituting~\eqref{eq:temp_0} into~\eqref{eq:fcfs_lb}, we have
\begin{equation}
\begin{aligned}
 \mathcal{L}_{FCFS}(\lambda) \geq &\underset{\omega:\omega<1-\lambda}{\max}\omega\underset{{{m}}\to\infty}{\lim} \frac{H\left({\mathbf{X}}^{{{*m}}}\big| {\mathbf{A}}^{{{*m}}}\right)-H\left({\mathbf{X}}^{{{*m}}}\big|{\mathbf{A}}^{{*m}}, q\left({A}^*_1\right),\cdots, q\left({A}^*_{{m}}\right)\right)}{{{m}}}.
\end{aligned}
\label{eq:temp_1}
\end{equation}

Applying the entropy chain rule to the second term in~\eqref{eq:temp_1},  we have
\begin{equation}
\begin{aligned}
\underset{{{m}}\to\infty}{\lim} \frac{H\left({\mathbf{X}}^{{{*m}}}\big|{\mathbf{A}}^{{*m}}, q\left({A}^*_1\right), \cdots, q\left({A}^*_{{m}}\right)\right)}{{{m}}}&\overset{}{=} \underset{{{m}}\to\infty}{\lim}  \frac{\sum_{k=1}^{{{m}}} H\left({X}^*_k\big|{\mathbf{X}}^{*k-1},  \mathbf{{A}}^{{{*m}}}, q\left({A}^*_1\right), \cdots, q\left({A}^*_{{{m}}}\right) \right)}{{{m}}}\\
&\overset{(d)}{=}  \underset{{{m}}\to\infty}{\lim} \frac{  \sum_{k=1}^{{{m}}} H\left({X}^*_k\big|{A}^*_{k}-{A}^*_{k-1}, q\left({A}^*_{k-1}\right), q\left({A}^*_{k}\right) \right)}{{{m}}},
\end{aligned}
\label{eq:fcfs_chain}
\end{equation}
where $(d)$ follows from the update equation of queue length seen by the attacker's jobs, which is given by
\begin{equation}
\begin{aligned}
q\left({A}^*_k\right) = \left(q\left({A}^*_{k-1}\right)+1+{X}^*_k-\left({A}^*_k-{A}^*_{k-1}\right) \right)_{+}.
\end{aligned}
\label{eq:fcfs_q_update}
\end{equation}

It can be shown that $\left\{{A}^*_{k}-{A}^*_{k-1}, q\left({A}^*_{k-1}\right), q\left({A}^*_{k}\right)\right\}$, $k\in\mathbb{Z}$, forms a positive recurrent Markov chain (See Appendix~\ref{app:fcfs_queue} for the proof), which implies the equivocation rate in~\eqref{eq:fcfs_chain} converges as $k\to\infty$, with the limit determined by the stationary distribution of $\left\{{A}^*_{k}-{A}^*_{k-1}, q\left({A}^*_{k-1}\right), q\left({A}^*_{k}\right)\right\}$. 
Let $\{\mathcal{T}, Q_1, Q_2\}$ take the stationary distribution of $\left\{{A}^*_{k}-{A}^*_{k-1}, q\left({A}^*_{k-1}\right), q\left({A}^*_{k}\right)\right\}$. From Ces\`{a}ro mean theorem \cite[Theorem~4.2.3]{Cover&Thomas:91}, \eqref{eq:fcfs_chain} can be rewritten as
\begin{equation}
\begin{aligned}
\underset{{{m}}\to\infty}{\lim} \frac{H\left({\mathbf{X}}^{{{*m}}}\big|{\mathbf{A}}^{{*m}}, q\left({A}^*_1\right), \cdots, q\left({A}^*_{{m}}\right)\right)}{{{m}}}& = H\left(\mathcal{X}|\mathcal{T}, Q_1, Q_2\right).
\label{eq:limit_1}
\end{aligned}
\end{equation}

Furthermore,  it can be shown that as $\omega\to1-\lambda$, $Q_2$ is always positive (See Appendix~\ref{app:eq_temp_4} for the proof), i.e.,
\begin{equation}
\begin{aligned}
 \underset{\omega\to 1-\lambda}{\lim} \mathbb{P}\left(Q_2 = 0\right) =0.
\label{eq:temp_4}
\end{aligned}
\end{equation} 
From the queue length update equation, 
$Q_2=\left(Q_1+1+\mathcal{X}-\mathcal{T}\right)_+
\label{eq:Q}$, we have that
\begin{equation}
\begin{aligned}
 H\left(\mathcal{X}|\mathcal{T}, Q_1, Q_2>0\right) = 0.
 \label{eq:temp_3}
 \end{aligned}
\end{equation}
\eqref{eq:temp_4} and~\eqref{eq:temp_3} imply that
\begin{equation}
\begin{aligned}
\underset{\omega\to 1-\lambda}{\lim} H\left(\mathcal{X}|\mathcal{T}, Q_1, Q_2\right) &=0.
\end{aligned}
\label{eq:temp_5}
\end{equation}

Substituting~\eqref{eq:limit_1} into~\eqref{eq:temp_1} and applying~\eqref{eq:temp_5}, we have 
\begin{equation}
\begin{aligned}
 \mathcal{L}_{FCFS}(\lambda)&\geq \underset{\omega:\omega<1-\lambda}{\max}
\omega \left(\underset{{{m}}\to\infty}{\lim} \frac{H\left({\mathbf{X}}^{{{*m}}}\big| {\mathbf{A}}^{{{*m}}}\right)}{{{m}}}- H\left(\mathcal{X}|\mathcal{T}, Q_1, Q_2\right)\right)\\
&\geq  \underset{\omega\to 1-\lambda}{\lim} \omega\left(\underset{{{m}}\to\infty}{\lim} \frac{H\left({\mathbf{X}}^{{{*m}}}\big| {\mathbf{A}}^{{{*m}}}\right)}{{{m}}} - H\left(\mathcal{X}|\mathcal{T}, Q_1, Q_2\right) \right) \\
&=  \underset{\omega\to 1-\lambda}{\lim}\omega\underset{{{m}}\to\infty}{\lim} \frac{H\left({\mathbf{X}}^{{{*m}}}\big| {\mathbf{A}}^{{{*m}}}\right)}{{{m}}} . 
\end{aligned}
\label{eq:temp_2}
\end{equation}

Since ${X}^*_i$'s are i.i.d. $Binomial({A}^*_{i}-{A}^*_{i-1},\lambda)$ random variables, we have
\begin{equation}
\begin{aligned}
\underset{{{m}}\to\infty}{\lim} \frac{H\left({\mathbf{X}}^{{{*m}}}\big| {\mathbf{A}}^{{{*m}}}\right)}{{{m}}} \overset{}{=}\alpha_\omega H\left(\sum_{i=1}^{\left\lfloor \frac{1}{\omega} \right\rfloor}\delta_i\right)  + (1-\alpha_\omega) H\left(\sum_{i=1}^{\left\lceil \frac{1}{\omega} \right\rceil}\delta_i\right),
\end{aligned}
\label{eq:limit_2}
\end{equation}
where $\alpha_\omega=\frac{\left\lceil \frac{1}{\omega} \right\rceil-\frac{1}{\omega}}{\left\lceil \frac{1}{\omega} \right\rceil-\left\lfloor \frac{1}{\omega} \right\rfloor}$.
Substituting~\eqref{eq:limit_2} into~\eqref{eq:temp_2} proves~\eqref{eq:achievability}. 
\end{proof}

\subsection{Proof of~\eqref{eq:temp_4}}
\label{app:eq_temp_4}

Recall that $\mathcal{T}$ is distributed as~\eqref{eq:attack_strategy_1},  $\mathcal{X}$ is $Binomial(\mathcal{T}, \lambda)$, and $Q_1$ and $Q_2$ have identical distribution satisfying
 \begin{equation}
\begin{aligned}Q_2=\left(Q_1+1+\mathcal{X}-\mathcal{T}\right)_+.
\label{eq:Q_app}
\end{aligned}
\end{equation}
We need to prove that $ \underset{\omega\to 1-\lambda}{\lim} \mathbbm{P}\left(Q_2 = 0\right) =0.
\label{eq:temp_4_app}$
\begin{proof}
First, \eqref{eq:Q_app} can be rewritten as
\begin{equation}\begin{aligned}
Q_2=\left(Q_1 + {Y} - \left\lceil \frac{1}{\omega} \right\rceil \right)_+,
\label{eq:temp_18_app}
\end{aligned}\end{equation}
where
\begin{equation}
{Y}=\begin{cases} 2+X_1  & \mbox{w.p. } \alpha_\omega \\ 1+X'_1 & \mbox{w.p. } 1-\alpha_\omega \end{cases},
\end{equation}
where $X_1\sim Binomial\left(\left\lfloor \frac{1}{\omega} \right\rfloor,\lambda\right)$, $X'_1\sim Binomial\left(\left\lceil \frac{1}{\omega} \right\rceil ,\lambda\right)$, and  $\alpha_\omega=\frac{\left\lceil \frac{1}{\omega} \right\rceil-\frac{1}{\omega}}{\left\lceil \frac{1}{\omega} \right\rceil-\left\lfloor \frac{1}{\omega} \right\rfloor}$.

Taking the $z$-transform of ${Q}_2$, we have
\begin{equation}
\begin{aligned}
\mathcal{Q}(z)=\frac{\sum_{k=0}^{\left\lceil \frac{1}{\omega} \right\rceil-2}\sum_{r=0}^{\left\lceil \frac{1}{\omega} \right\rceil-2-k}p_ku_r\left(z^{\left\lceil \frac{1}{\omega} \right\rceil-1}-z^{k+r}\right)}{z^{\left\lceil \frac{1}{\omega} \right\rceil-1}-\mathcal{Y}(z)}
\label{eq:q_z_app}
\end{aligned}
\end{equation}
where $\mathcal{Y}(z)$ is the $z$-transform of $Y$ as given by
\begin{equation}
\begin{aligned}
\mathcal{Y}(z) = \alpha_{\omega}z(1-\lambda+\lambda z)^{\left\lceil\frac{1}{\omega}\right\rceil} +
(1-\alpha_{\omega})z^2(1-\lambda+\lambda z)^{\left\lfloor\frac{1}{\omega}\right\rfloor},
\label{eq:y_z_app}
\end{aligned}
\end{equation}
and $p_k=\mathbbm{P}\left(Q_2=k\right)$ and $u_r=\mathbbm{P}\left(Y=r\right)$.

Substituting~\eqref{eq:y_z_app} into~\eqref{eq:q_z_app} and letting $z=1$ on both sides, we have
\begin{equation}
\begin{aligned}
\sum_{k=0}^{\left\lceil \frac{1}{\omega} \right\rceil-2}\left(\sum_{r=0}^{\left\lceil \frac{1}{\omega} \right\rceil-2-r}u_r\left(\left\lceil \frac{1}{\omega} \right\rceil-1-k-r\right)\right) = \left\lceil \frac{1}{\omega} \right\rceil \left(1-\omega-\lambda\right).
\end{aligned}
\end{equation}
Dropping the terms with $r>1$ on the left hand side, we further get 
\begin{equation}
\begin{aligned}
u_1\left(\left\lceil \frac{1}{\omega} \right\rceil -1 -k\right)\sum_{k=0}^{\left\lceil \frac{1}{\omega} \right\rceil-2}p_k\leq \left\lceil \frac{1}{\omega} \right\rceil\left(1-\omega-\lambda\right).
\end{aligned}
\end{equation}
Plugging in the values of $u_1$ and taking the limit as $\omega\to 1-\lambda$, we have
\begin{equation}
\begin{aligned}
\underset{\omega\to1-\lambda}{\lim} \sum_{k=0}^{\left\lceil \frac{1}{\omega} \right\rceil-2} p_k \leq \underset{\omega\to1-\lambda}{\lim}  \frac{\left\lceil \frac{1}{\omega} \right\rceil\left(1-\omega-\lambda\right)}{(1-\alpha_\omega) \left(1-\lambda^{\left\lceil \frac{1}{\omega} \right\rceil} \right)} =0,
\end{aligned}
\end{equation}
which readily implies $p_0=0$, which is the desired result. 
\end{proof}

\subsection{Queuing Analysis of the FCFS scheduler Under a Periodic Sampling Attack}
\label{app:fcfs_queue}

\begin{theorem}
In an FCFS scheduler with total job arrival rate below 1, when the attacker   applies the periodic-sampling strategy defined in~\eqref{eq:attack_strategy_1}, 
the tuples $\left\{{A}^*_{k}-{A}^*_{k-1}, q\left({A}^*_{k-1}\right), q\left({A}^*_{k}\right)\right\}$, $k\in\mathbb{Z}$, form a positive recurrent Markov chain.
\label{theo:app_markov}
\end{theorem}

\begin{proof}
We first prove that $\left\{q\left({A}^*_{k}\right)\right\}$, $k\in \mathbb{Z}$, form a positive recurrent Markov chain. The Markovian property directly follows from the FCFS policy and memoryless property of the user's arrival process; given the queue length at ${A^*}_{k}$, future queue states are independent with the past arrival history. 

We show the ergodicity using a linear Lyapunov function defined as
\begin{equation}
\begin{aligned}
V\left(q\left({A}^*_k\right)\right) = q\left({A}^*_k\right), k\in\mathbb{Z}.
\end{aligned}
\end{equation}

Recall that the arrival rates of the user and attacker by $\lambda$ and $\omega$, respectively. 
If $q\left({A}^*_k\right)\geq  \left\lceil \frac{1}{\omega} \right\rceil $, the scheduler is guaranteed to be busy from ${A}^*_{k}$ to ${A}^*_{k+1}$. Hence, from~\eqref{eq:fcfs_q_update},
\begin{equation}
\begin{aligned}
q\left({A^*}_{k+1}\right) = 
q\left({A^*}_k\right) + 1 + X^*_k -\left({A}^*_{k+1}-{A}^*_k\right). 
\end{aligned}
\end{equation}
As $X^*_k$ has mean as $\frac{\lambda}{\omega}$, and ${A}^*_{k+1}-{A}^*_k$ has mean as $\frac{1}{\omega}$,
 the drift of the Lyapunov function in this case is given by 
\begin{equation}
\begin{aligned}
\mathbb{P}V\left(q\left({A}^*_k\right)\right) -  V\left(q\left({A}^*_k\right)\right) = -\frac{1-\omega-\lambda}{\omega}. 
\end{aligned}
\label{eq:lyapunov_1}\end{equation}
Additionally, during $[{A}^*_k,{A}^*_{k+1})$, the buffer queue length can grow at most 1, so the drift is bounded by
\begin{equation}
\begin{aligned}
\mathbb{P}V\left(q\left({A}^*_k\right)\right) -  V\left(q\left({A}^*_k\right)\right) \leq 1.  
\end{aligned}
\label{eq:lyapunov_2}\end{equation}
Combining~\eqref{eq:lyapunov_1} and~\eqref{eq:lyapunov_2},  the drift in any queue state is bound by 
\begin{equation}
\begin{aligned}
\mathbb{P}V\left(q\left({A}^*_k\right)\right) -  V\left(q\left({A}^*_k\right)\right)  \leq -\epsilon +\mathbbm{1}_{\left\{q\left({A}^*_k\right)< \left\lceil \frac{1}{\omega} \right\rceil\right\}},
\label{eq:app_temp_1}
\end{aligned}
\end{equation}
where $\epsilon = \frac{1-\omega-\lambda}{\omega}$. Following Foster-Lyapunov stability~\cite[Theorem~5]{Foster}, \eqref{eq:app_temp_1} implies the Markov chain $\left\{q\left({A}^*_{k}\right)\right\}$, $k\in\mathbb{Z}$, is positive recurrent. 

For the same reason we argue for the Markovian property of chain $\left\{q\left({A}^*_{k}\right)\right\}$, $k\in\mathbb{Z}$, $\left\{{A}^*_{k}-{A}^*_{k-1}, q\left({A}^*_{k-1}\right), q\left({A}^*_{k}\right)\right\}$, $k\in\mathbb{Z}$,  also form a Markov chain. Additionally,  a stationary state distribution of $\left\{{A}^*_{k}-{A}^*_{k-1}, q\left({A}^*_{k-1}\right), q\left({A}^*_{k}\right)\right\}$, $k\in\mathbb{Z}$, can be derived from the stationary state distribution of $\left\{q\left({A}^*_{k}\right)\right\}$ and~\eqref{eq:attack_strategy_1}. The existence of the stationary distribution implies that $\left\{{A}^*_{k}-{A}^*_{k-1}, q\left({A}^*_{k-1}\right), q\left({A}^*_{k}\right)\right\}$, $k\in\mathbb{Z}$,  must be positive recurrent~\cite[Definition~3.1]{gilks95}. 
\end{proof}

\subsection{Busy Period Distribution of the Round Robin Scheduler}
\label{app:rr_busy}

Consider a round robin scheduler serving a user and an attacker.
The user sends jobs according to a Bernoulli arrival process with rate $\lambda\leq0.5$, and the attacker applies the nonstop monitoring attack, 
where the arrival and departure times satisfy~\eqref{eq:nonstop}. 
We prove the busy periods seen by the attacker, $B_r$'s as defined in~\eqref{eq:rr_busy_def}, are i.i.d. distributed as~\eqref{eq:rr_busy_1} and have mean 
\begin{equation}
\begin{aligned}
\mathbb{E}\left[B_r\right] = \underset{n\to\infty}{\lim}\frac{\sum_{r=1}^{n}B_r}{n}= \frac{1}{1-2\lambda}.
\label{eq:rr_busy_mean_app}
\end{aligned}
\end{equation}
\begin{proof}
Write the update equation of queue lengths seen by the attacker as 
\begin{equation}
\begin{aligned}
q\left({A}_{k+1}\right) = \left( q\left({A}_{k}\right) +1+\sum_{i={A}_k+1}^{{A}_{k+1}}\delta_i - \left({A}_{k+1}-{A}_{k}\right)  \right)_+,
\label{eq:app_temp_12_1}
\end{aligned}
\end{equation}
for $k\in\mathbb{Z}$. 

From~\eqref{eq:nonstop} and~\eqref{eq:rr-nonstop}, we have
\begin{equation}
\begin{aligned}
{A}_{k+1}-{A}_{k} = \begin{cases} 1 \quad \text{ if } q\left({A}_k\right)=0\\ 2 \quad \text{ if } q\left({A}_k\right)>0\end{cases}.
\label{eq:app_temp_13}
\end{aligned}
\end{equation}

Given~\eqref{eq:app_temp_12_1} and~\eqref{eq:app_temp_13}, we can draw a Markov chain formed by $\left\{q\left({A}_{k}\right)\right\}$, $k\in\mathbb{Z}$, as depicted by Figure~\ref{fig:rr_busy_markov}. 
The length of busy period $B_r$ is simply a  function of the number of transitions, $s$, it takes to return back to state $0$ (starting from state $0$), as given by
\begin{equation}
B_r=2s-1, \quad r\in\mathbb{Z}.
\label{eq:app_temp_14}
\end{equation}
Clearly, $s$ has the same PMF as $B_r$'s.

\begin{figure}[t]
   \centering
   \includegraphics[width=0.5\columnwidth]{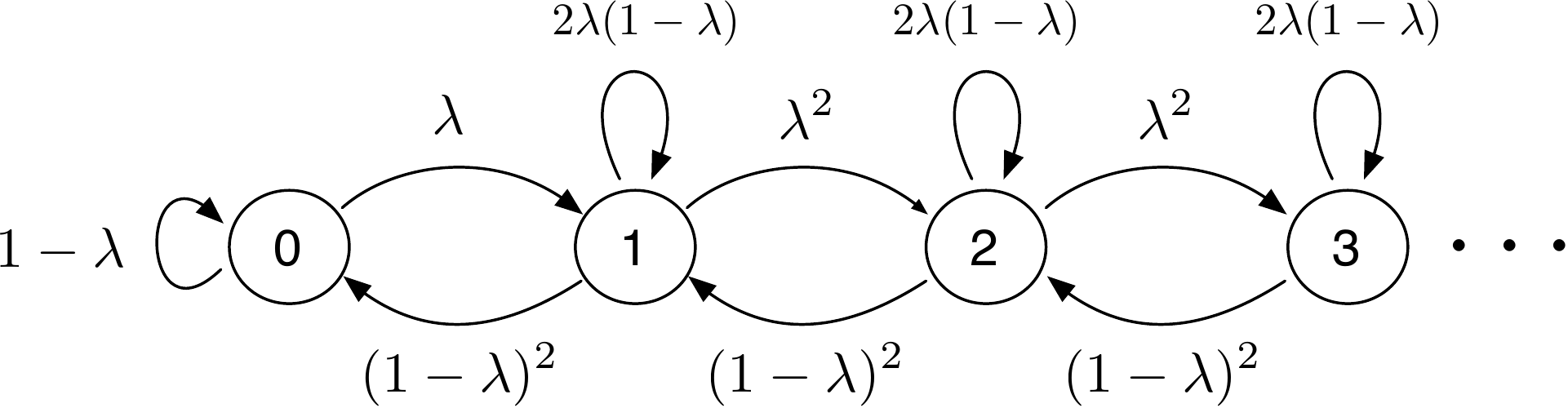}
   \caption{The Markov chain of queue length when attacking the round robin scheduler with the nonstop monitoring strategy.}
   \label{fig:rr_busy_markov}
\end{figure} 

From the Markov chain in Figure~\ref{fig:rr_busy_markov},  we know
\begin{equation}\mathbb{P}(s=1)=1-\lambda.
\label{eq:app_temp_15}
\end{equation} 
For $s=2$, the queue length first needs one step to jump to state 1, and then returns to state 0, which has the probability as  \begin{equation}
\mathbb{P}(s=2)=\lambda(1-\lambda)^2.\label{eq:app_temp_16}\end{equation} 
For $s>2$, after jumping to state 1 for first time, the queue state needs to experience another $s-2$ transitions before returning to state 1 for the last time and eventually returning back to state 0. 
Notice that each state transition either increases the queue length by 1, decreases the queue length by 1, or keeps the queue length unchanged. 
Therefore, the $s-2$ intermediate transitions must contain the same number of transitions that increase and decrease the queue length.
Moreover, the increase in the queue length must always lead the decrease; otherwise, the queue would return to state 0 before the require number of transitions finish.  
Define $W_{s,j}$ to be the total ways of 
 $s-2$ transitions with $j$ transitions  increasing the queue (also $j$ transitions decreasing the queue).
 Then $W_{s,j}=\binom{s-2}{2j}C_j$,  where $C_j$ is the famous Catalan number as given by
\begin{equation}
\begin{aligned}
C_j=\frac{1}{j+1}\binom{2j}{j}. 
\end{aligned}
\end{equation} 
The transition probability for $s>2$ can be calculated as
\begin{equation}
\begin{aligned}
\mathbb{P}(s=r) &= \lambda (1-\lambda)^2\sum_{j=1}^{\left\lfloor\frac{r-2}{2}\right\rfloor}W_{r,j} \left(2(1-\lambda)\lambda\right)^{r-2-2j}(1-\lambda)^{2j}\lambda^{2j}\\
 &= 2^{r-1}\lambda^{r-1}(1-\lambda)^{r} \sum_{j=1}^{\left\lfloor\frac{r-2}{2}\right\rfloor} \frac{(r-2)!2^{-2j-1}}{(r-2-2j)!j!(j+1)!}.
\end{aligned}
\label{eq:app_temp_17}
\end{equation}
\eqref{eq:app_temp_14}, \eqref{eq:app_temp_15}, \eqref{eq:app_temp_16}, and~\eqref{eq:app_temp_17} together imply~\eqref{eq:rr_busy_1}.

In each busy period, the number of queue state transitions equals to the number of attacker's  job arrivals.
Since the attacker's arrival rate in this nonstop-monitoring attack is $1-\lambda$, thus $\mathbb{E}[\frac{s}{B_r}]={1-\lambda}$, which together with~\eqref{eq:app_temp_14} prove~\eqref{eq:rr_busy_mean_app}. 
\end{proof}

\subsection{Busy Period Distribution of the WC-TDMA Scheduler}
\label{app:nitdma_busy}

Consider a WC-TDMA scheduler serving a user and an attacker, where even and odd time slots are reserved to the attacker and user, respectively.
The user sends jobs according to Bernoulli arrival process with rate $\lambda\leq0.5$, and the attacker sends jobs on all odd slots
according to~\eqref{eq:half_attack}. 
We prove the busy periods seen by the attacker, $B'_r$'s as defined in~\eqref{eq:nitdma_busy_def}, are i.i.d. distributed as~\eqref{eq:rr_busy_1} and have mean 
\begin{equation}
\mathbb{E}\left[{B}'_r\right] = \frac{2-2\lambda}{1-2\lambda}, \quad \text{for } r\in\mathbb{Z}.
\label{eq:nitdma_busy_mean_app}
\end{equation}
\begin{proof}
%
Based on~\eqref{eq:half_attack} and~\eqref{eq:app_temp_12_1}, the evolving of queue state, $\left\{q\left({A}_{k}\right)\right\}$, $k\in\mathbb{Z}$,  can be described by the same Markov chain in Figure~\ref{fig:rr_busy_markov}, and the busy period $B'_i$ is nothing but  a function of the number of transitions,  $s$, it takes to  return to state 0 starting from state 0, as given by
\begin{equation}
B'_r=2s, \quad r\in\mathbb{Z}.
\label{eq:app_temp_19}
\end{equation}

Using the same arguments in our proof of~\eqref{eq:rr_busy_1} in Appendix~\ref{app:rr_busy},  
it is clear that $B'_r$'s are distributed as~\eqref{eq:rr_busy_1}.
Additionally,  from~\eqref{eq:app_temp_14}, ~\eqref{eq:app_temp_19}, and~\eqref{eq:rr_busy_mean_app},
\begin{equation}
\begin{aligned}
\mathbb{E}\left[B'_r\right] & = \mathbb{E}\left[B_r\right]+1
=\frac{2-2\lambda}{1-2\lambda}.
\end{aligned}
\end{equation}
\end{proof}

\subsection{Proof of~\eqref{eq:temp_13}}
\label{app:eq_temp_13}
We need to prove that
\begin{equation}
\begin{aligned}
\underset{\lambda\to0}{\lim} \text{ }  \mathbb{P}\left(q\left(S_k\right) =0\right) = 1, \quad k\in\mathbb{Z},
\end{aligned}
\label{eq:temp_13_app}
\end{equation}
where $S_k = \max\{A_k,D_{k-1}\}$.

\begin{proof}
We prove this by induction. 
The base case is straightforward, considering the queue length starts with zero, i.e., $q\left(S_1\right)=0$.
Now assume it is true that $\underset{\lambda\to0}{\lim}\text{ }   \mathbb{P}\left(q(S_k) =0\right) = 1$, $\text{for } k=1, 2, \cdots, r$. Consider 
the update equation of queue length, as given by
\begin{equation}
\begin{aligned}
q\left(S_{r+1}\right) = \left(q\left(S_r\right)+1+\sum_{i=S_r}^{S_{r+1}-1}\delta_i-\left(S_{r+1}-S_r\right)\right)_+,
\label{eq:temp_13_app_1}
\end{aligned}
\end{equation}
from which we calculate the probability of empty queue as given by
\begin{equation}
\begin{aligned}
&\underset{\lambda\to0}{\lim}  \text{ } \mathbb{P}\left(q\left(S_{r+1}\right) =0\right)\\
& =  \underset{\lambda\to0}{\lim}  \sum_{i=0}^{S_{r+1}-S_{r}-1}\mathbb{P}\left(q\left(S_{r}\right) =i\right)\cdot\sum_{j=0}^{S_{r+1}-S_{r}-1-i} \binom{S_{r+1}-S_{r}}{j}(1-\lambda)^{S_{r+1}-S_{r}-j}\lambda^{j}\\
&\overset{(a)}{=}  \underset{\lambda\to0}{\lim}  \sum_{j=0}^{S_{r+1}-S_{r}-1} \binom{S_{r+1}-S_{r}}{j}(1-\lambda)^{S_{r+1}-S_{r}-j}\lambda^{j}\\
&=   1-\underset{\lambda\to0}{\lim}  \lambda^{S_{r+1}-S_{r}} \\
&=1,
\end{aligned}
\label{eq:temp_13_app_2}
\end{equation}
where $(a)$ follows from the assumption that  $\underset{\lambda\to0}{\lim}\text{ }\mathbb{P}\left(q(S_k) =0\right) = 1$. 
This completes the proof.
\end{proof}

\bibliographystyle{IEEETran}
\bibliography{thesis}

\begin{thebibliography}{10}
\providecommand{\url}[1]{#1}
\csname url@samestyle\endcsname
\providecommand{\newblock}{\relax}
\providecommand{\bibinfo}[2]{#2}
\providecommand{\BIBentrySTDinterwordspacing}{\spaceskip=0pt\relax}
\providecommand{\BIBentryALTinterwordstretchfactor}{4}
\providecommand{\BIBentryALTinterwordspacing}{\spaceskip=\fontdimen2\font plus
\BIBentryALTinterwordstretchfactor\fontdimen3\font minus
  \fontdimen4\font\relax}
\providecommand{\BIBforeignlanguage}[2]{{%
\expandafter\ifx\csname l@#1\endcsname\relax
\typeout{** WARNING: IEEEtran.bst: No hyphenation pattern has been}%
\typeout{** loaded for the language `#1'. Using the pattern for}%
\typeout{** the default language instead.}%
\else
\language=\csname l@#1\endcsname
\fi
#2}}
\providecommand{\BIBdecl}{\relax}
\BIBdecl

\bibitem{lampson73:acm}
B.~W. Lampson, ``A note on the confinement problem,'' \emph{Commun. ACM},
  vol.~16, no.~10, pp. 613--615, October 1973.

\bibitem{millen89}
J.~K. Millen, ``Finite-state noiseless covert channels,'' in \emph{Computer
  Security Foundations Workshop}, Franconia, NH, 1989, pp. 81--86.

\bibitem{Cabuk2009}
S.~Cabuk, C.~E. Brodley, and C.~Shields, ``{IP} covert channel detection,''
  \emph{ACM Transactions on Information and System Security (TISSEC)}, vol.~12,
  no.~4, 2009.

\bibitem{ristenpart09:ccs}
T.~Ristenpart, E.~Tromer, H.~Shacham, and S.~Savage, ``{Hey, you, get off of my
  cloud: exploring information leakage in third-party compute clouds},'' in
  \emph{ACM Conf. on Computer and Communications Security (CCS)}, Chicago, IL,
  2009, pp. 199--212.

\bibitem{gong2012pets}
X.~Gong, N.~Borisov, N.~Kiyavash, and N.~Schear, ``Website detection using
  remote traffic analysis,'' in \emph{Privacy Enhancing Technologies (PETS)},
  Vigo, Spain, 2012, pp. 58--78.

\bibitem{kadloor13}
S.~Kadloor and N.~Kiyavash, ``Delay optimal policies offer very little
  privacy,'' in \emph{IEEE International Conf. on Computer Communications
  (Infocom)}, Turin, Italy, 2013, pp. 2454--2462.

\bibitem{peterson2007computer}
L.~L. Peterson and B.~S. Davie, \emph{Computer networks: a systems
  approach}.\hskip 1em plus 0.5em minus 0.4em\relax Elsevier, 2007.

\bibitem{AnantharamVerdu96}
V.~{}Anantharam and S.~{}Verd\'{u}, ``Bits through queues,'' \emph{IEEE Trans.
  on Inf. Theory}, vol.~42, no.~1, pp. 4--18, 1996.

\bibitem{sellke2007capacity}
S.~H. Sellke, C.-C. Wang, N.~Shroff, and S.~Bagchi, ``Capacity bounds on timing
  channels with bounded service times,'' in \emph{IEEE International Symposium
  on Inf. Theory}, Nice, France, 2007, pp. 981--985.

\bibitem{riedl2011finite}
T.~J. Riedl, T.~P. Coleman, and A.~C. Singer, ``Finite block-length achievable
  rates for queuing timing channels,'' in \emph{IEEE Information Theory
  Workshop (ITW)}, Paraty, Brazil, 2011, pp. 200--204.

\bibitem{Siva}
S.~Gorantla, S.~Kadloor, T.~Coleman, N.~Kiyavash, I.~Moskowitz, and M.~Kang,
  ``Characterizing the efficacy of the {NRL} network pump in mitigating covert
  timing channels,'' \emph{IEEE Trans. on Inf. Forensics and Security}, vol.~7,
  no.~1, pp. 64--75, 2012.

\bibitem{Giles02}
J.~Giles and B.~Hajek, ``An information-theoretic and game-theoretic study of
  timing channels,'' \emph{IEEE Trans. on Inf. Theory}, vol.~48, pp.
  2455--2477, 2002.

\bibitem{askarov10:ccs}
A.~Askarov, D.~Zhang, and A.~C. Myers, ``Predictive black-box mitigation of
  timing channels,'' in \emph{ACM conference on Computer and communications
  security}, Chicago, IL, 2010, pp. 297--307.

\bibitem{zhang11:ccs}
D.~Zhang, A.~Askarov, and A.~C. Myers, ``Predictive mitigation of timing
  channels in interactive systems,'' in \emph{ACM conference on Computer and
  communications security}, Chicago, IL, 2011, pp. 563--574.

\bibitem{Murdoch2005}
S.~Murdoch and G.~Danezis, ``Low-cost traffic analysis of {Tor},'' in
  \emph{IEEE Symposium on Security and Privacy}, V.~Paxson and M.~Waidner,
  Eds., Berkeley, CA, May 2005, pp. 183--195.

\bibitem{Dingledine2004}
R.~Dingledine, N.~Mathewson, and P.~Syverson, ``Tor: The second-generation
  onion router,'' in \emph{USENIX Security Symposium}, M.~Blaze, Ed., San
  Diego, CA, 2004, pp. 303--320.

\bibitem{morphmix}
M.~Rennhard and B.~Plattner, ``Introducing {MorphMix}: peer-to-peer based
  anonymous internet usage with collusion detection,'' in \emph{ACM Workshop on
  Privacy in Electronic Society}, New York, NY, 2002, pp. 91--102.

\bibitem{kiyavash2013timing}
N.~Kiyavash, F.~Koushanfar, T.~P. Coleman, and M.~Rodrigues, ``A timing channel
  spyware for the {CSMA/CA} protocol,'' \emph{IEEE Trans. on Inf. Forensics and
  Security}, vol.~8, no.~3, pp. 477--487, 2013.

\bibitem{Sachin10}
S.~Kadloor, X.~Gong, N.~Kiyavash, T.~Tezcan, and N.~Borisov, ``Low-cost side
  channel remote traffic analysis attack in packet networks,'' in \emph{IEEE
  International Conference on Communications}, C.~Xiao and J.~C. Olivier, Eds.,
  Cape Town, South Africa, 2010.

\bibitem{gong2011isit}
X.~Gong, N.~Kiyavash, and P.~Venkitasubramaniam, ``Information theoretic
  analysis of side channel information leakage in fcfs schedulers,'' in
  \emph{IEEE International Symposium on Information Theory (ISIT)},
  Saint-Petersburg, Russia, 2011, pp. 1255--1259.

\bibitem{Wyner75}
A.~D. Wyner, ``The wire-tap channel,'' \emph{Bell Sys. Tech. J.}, vol.~54, pp.
  1355--1387, 1975.

\bibitem{shakkottai02}
S.~Shakkottai and A.~L. Stolyar, ``Scheduling for multiple flows sharing a
  time-varying channel: The exponential rule,'' \emph{Translations of the
  American Mathematical Society-Series 2}, vol. 207, pp. 185--202, 2002.

\bibitem{gail93}
H.~R. Gail, G.~Grover, R.~Gu{\'e}rin, S.~L. Hantler, Z.~Rosberg, and M.~Sidi,
  ``Buffer size requirements under longest queue first,'' \emph{Performance
  Evaluation}, vol.~18, no.~2, pp. 133--140, 1993.

\bibitem{Cover&Thomas:91}
T.~M. Cover and J.~A. Thomas, \emph{{Elements of Information Theory}}.\hskip
  1em plus 0.5em minus 0.4em\relax New York: Wiley, 1987.

\bibitem{hahne86}
E.~L. Hahne and R.~G. Gallager, ``Round robin scheduling for fair flow control
  in data communication networks,'' \emph{NASA STI/Recon Technical Report N},
  vol.~86, p. 30047, 1986.

\bibitem{siegele2008let}
L.~Siegele, \emph{Let it rise: A special report on corporate IT}.\hskip 1em
  plus 0.5em minus 0.4em\relax Economist Newspaper, 2008.

\bibitem{bruneel86}
H.~Bruneel, ``Message delay in tdma channels with contiguous output,''
  \emph{IEEE Trans. on Communication}, vol.~34, no.~7, pp. 681--684, 1986.

\bibitem{Murota}
K.~Murota and A.~Shioura, ``{Relationship of M-/L-convex functions with
  discrete convex functions by Miller and Favati-Tardella},'' \emph{Discrete
  Applied Mathematics}, vol. 115, pp. 151--176, 2001.

\bibitem{krantz95}
S.~G. Krantz, \emph{{Handbook of Complex Variables}}.\hskip 1em plus 0.5em
  minus 0.4em\relax MA: BirkhŠuser, 1995.

\bibitem{Foster}
F.~G. Foster, ``On the stochastic matrices associated with certain queuing
  processes,'' \emph{Ann. Math. Statistics}, vol.~24, pp. 355--360, 1953.

\bibitem{gilks95}
W.~Gilks, S.~Richardson, and D.~Spiegelhalter, \emph{{Markov Chain Monte Carlo
  in Practice}}.\hskip 1em plus 0.5em minus 0.4em\relax Chapman and Hall, 1995.

\end{thebibliography}

\end{document}